\documentclass[12pt]{article}

\usepackage{ulem}
\usepackage{graphicx}
\usepackage{amsfonts}
\usepackage{amscd}
\usepackage{indentfirst}
\usepackage{amssymb,amsthm, amsgen, amstext, amsbsy, amsopn,yfonts}
\usepackage{xcolor}
\usepackage{amsmath}
\usepackage{mathrsfs}
\usepackage[hyphens]{url}

\newtheorem{thm}{Theorem}[section]

\newtheorem{prop}[thm]{Proposition}

\newtheorem{defin}[thm]{Definition}
\newtheorem{rem}[thm]{Remark}
\newtheorem{exam}[thm]{Example}

\newtheorem{question}[thm]{Question}

\newcommand{\C}{{\mathbb{C}}}

\newcommand{\N}{{\mathbb{N}}}

\newcommand{\R}{{\mathbb{R}}}
\newcommand{\SP}{{\mathbb{S}}}

\newcommand{\Z}{{\mathbb{Z}}}

\newcommand{\cH}{{\mathcal{H}}}

\newcommand{\cS}{{\mathcal{S}}}

\newcommand{\cU}{{\mathcal{U}}}
\newcommand{\cV}{{\mathcal{V}}}

\newcommand{\Ker}{{\rm Ker\,}}

\newcommand{\Qed}{\ \hfill \qedsymbol \bigskip}

\usepackage{lineno}
\begin{document}

\title{Contact topology and non-equilibrium thermodynamics}

\author{\textsc Michael Entov$^{1}$,\ Leonid Polterovich$^{2,3}$ }

\footnotetext[1]{Partially supported by the Israel Science Foundation grant 1715/18.}

\footnotetext[2]{ Partially supported by the Israel Science Foundation grant 1102/20.}

\footnotetext[3]{Corresponding author.}

\date{\today}

\maketitle

\begin{abstract} We describe a method, based on contact topology, of showing the existence of semi-infinite trajectories of contact Hamiltonian flows which start on one Legendrian submanifold and asymptotically converge to another Legendrian submanifold. We discuss a mathematical model of non-equilibrium thermodynamics where such trajectories play a role of relaxation processes, and illustrate our results in the case of the Glauber dynamics for the mean field Ising model.
\end{abstract}

\tableofcontents


\section{Introduction and outline}
The goal of the present paper is twofold. First, we describe a method, based on
``hard" contact topology, of {showing the existence of} semi-infinite trajectories of contact Hamiltonian flows which start on one Legendrian submanifold and asymptotically converge to another Legendrian submanifold. Second, we discuss a mathematical model of non-equilibrium thermodynamics where such trajectories play a role of relaxation processes, and illustrate
our results in the case of Glauber dynamics for the Ising model in the mean field approximation. Our starting point is the existence mechanism, called {\it interlinking},  for finite time-length trajectories between a pair of Legendrians developed in \cite{EP-Leg-pers-mod}.

\subsection{{Asymptotic trajectories of contact flows}}

{Let us briefly review a few preliminaries from contact geometry.}
Recall \cite{Geiges} that a (cooriented) {\it contact structure} on an odd-dimensional manifold $\Sigma^{2n+1}$ is a  field of tangent hyperplanes $\xi \subset T\Sigma$ forming
the kernel of a $1$-form $\lambda$ on $\Sigma$ with $\lambda \wedge (d\lambda)^n$ being a volume form on $\Sigma$. Note that $\lambda$ is not canonical: it is defined up to multiplication by a positive function. \color{black}{We shall denote a contact manifold
either as $(\Sigma,\xi)$ or, when we wish to highlight a specific contact form, as $(\Sigma, \lambda)$. \color{black} Every choice of the $1$-form $\lambda$ defines the {\it Reeb} vector field $R$ on $\Sigma$ by $i_R d\lambda =0$ and $\lambda(R)=1$.
 The flow of $R$ is called {\it the Reeb flow}.
An $n$-dimensional submanifold of $\Sigma$ is called {\it Legendrian} if it is everywhere tangent to $\xi$.  Finite pieces
of trajectories of the Reeb flow starting and ending at given Legendrian submanifolds
are called {\it Reeb chords}.

A vector field $v$  on $\Sigma$ preserving the contact structure $\xi$ is called a {\it contact} vector field. Given a contact form $\lambda$, {\it the contact Hamiltonian}
of a contact vector field $v$ is the function $H:= \lambda(v)$.
One can show that $H$ uniquely determines such a $v$.
For instance, the contact Hamiltonian of the Reeb vector field of $\lambda$ equals $1$ everywhere on $\Sigma$. On the other hand, the contact vector field $v$ generated by any strictly positive Hamiltonian $H$ is the Reeb vector field of $\lambda/H$.

\begin{exam}
\label{exam-onejet} {\rm Let $\Sigma = J^1 X = T^* X\times \R (z)$ be the 1-jet space of a smooth manifold $X$, together with the standard contact form $dz-pdq$ on it. Here $(p,q)$ are the canonical coordinates on $T^*X$. The Reeb vector field
is simply $\partial/\partial z$, and its flow is the shift in the $z$-coordinate. Given a smooth function $\phi$ on $X$, its $1$-jet map in $(p,q,z)$ - coordinates is given by
$$j^1\phi: X \to \Sigma, \;\; x \mapsto \left(\frac{\partial \phi}{\partial x}, x, \phi(x)\right)\;.$$ Its image $\Lambda= j^1\phi(X)$ is a Legendrian submanifold.
For $\phi=0$, $\Lambda$ is {\it the zero section} of $J^1X$.

The Hamiltonian vector field of a function $H(p,q,z)$
is given in $(p,q,z)$ - coordinates by
$$\left(\frac{\partial H}{\partial q} +p \frac{\partial H}{\partial z},\;-\frac{\partial H}{\partial p}, H -p\frac{\partial H}{\partial p}\right)\;.$$}
\end{exam}

\bigskip

Starting with the famous Arnold's conjecture on Reeb chords, the existence of Reeb chords of Legendrian submanifolds (under appropriate assumptions on the submanifolds and the ambient contact manifold) has been one of the central problems of contact dynamics. The following phenomenon, called {\it interlinking} (see Section~\ref{subsec-interlink} for precise definitions), is of a particular importance for the present paper: there exist pairs of disjoint Legendrians $\Lambda_0, \Lambda_1$ such that every Reeb flow (i.e., the flow of a strictly positive contact Hamiltonian) possesses a chord of controlled time-length starting on $\Lambda_0$ and ending on $\Lambda_1$. Moreover, in certain cases this property of $\Lambda_0, \Lambda_1$ holds for any pair $\Lambda'_0, \Lambda'_1$ obtained from $\Lambda_0, \Lambda_1$ by a sufficiently small (say, $C^\infty$-small) Legendrian isotopy -- in such a case we will say that the pair $(\Lambda_0, \Lambda_1)$ is {\it robustly interlinked}.

It is worth mentioning that interlinking is detected in \cite{EP-Leg-pers-mod} with the help of Legendrian Contact Homology (LCH)  \cite{EGH}, a sophisticated algebraic structure associated to a Legendrian link and involving ideas coming from the string theory. The key ingredient of LCH is a count
of special
non-compact
surfaces in $\Sigma \times \R$ whose boundaries lie on $\Lambda_0 \times \R$
and $\Lambda_1 \times
 \R
$
and which are asymptotic to a collection of the Reeb chords.
These surfaces are interpreted as world-sheets of open strings.
While LCH remains invisible in this paper, it would not be a stretch to say that our journey begins there and eventually leads, through contact topology, to non-equilibrium thermodynamics. We refer the reader to \cite{Morrison} for a survey of string theoretic aspects of symplectic topology, and to \cite{Aga} for interactions  between string theory
and Legendrian contact homology.

The first main objective of the present paper is to show that given an
interlinked pair $(\Lambda_0,\Lambda_1)$, certain contact Hamiltonians that vanish on $\Lambda_1$ and may change sign elsewhere, have semi-infinite trajectories starting on $\Lambda_0$  and {\it asymptotically converging to $\Lambda_1$, as $t\to +\infty$} (meaning that $d(\gamma(t),\Lambda_1)\to 0$ as $t\to +\infty$, where $d$ is a distance function defined by a Riemannian metric on $\Sigma$). Let us state here, rather informally, our main result in this direction and refer to Theorem~\ref{thm-A}  for the precise formulation, and to Theorem~\ref{thm-semimain} for a generalization.

Assume $H:\Sigma\to\R$ is a contact Hamiltonian and $v$ is its contact vector field $v$, so that the flow $\{\varphi^t\}$ of $v$ is defined for all times.

\bigskip
\noindent
{\bf ``Theorem" (an informal version of Theorem~\ref{thm-A} below).}
{\it Let a submanifold $M \subset \Sigma$ be the union of a finite number of the connected components of the
nodal set $\{H=0\}$ which separates $\Sigma$ into two open parts, $\Sigma_-$ and $\Sigma_+$.
Assume that $H$ is strictly positive on $\Sigma_+$ (but may change sign on $\Sigma_-$).
Let  $(\Lambda_0, \Lambda_1)$ be a robustly interlinked pair of Legendrian submanifolds
with $\Lambda_1 \subset M$.  Then, under certain dynamical assumptions concerning the contact flow $\{\varphi^t\}$ of $H$ near $\Lambda_1$ and on $\Sigma_+$, we have $\Lambda_0 \cap \Sigma_+ \neq \emptyset$, and there exists a trajectory of the flow, lying completely in $\Sigma_+$, which starts on $\Lambda_0$ and asymptotically converges to $\Lambda_1$.}

\medskip
\begin{exam} \label{exam-attjet} {\rm In the notations of Example~\ref{exam-onejet},
the zero section $\Lambda_1 \subset M:= \{z=0\}$ is the global attractor of the contact Hamiltonian flow generated by $-cz$, $c>0$
(see Section~\ref{subsec-attractors-repellers} for a precise definition of a global attractor).
This can be seen by solving the corresponding
Hamiltonian system $\dot{z}=-cz, \dot{p}=-cp$. Methods of contact topology enable us to
detect semi-infinite trajectories starting on $\Lambda_0 =\{p=0,z=-1\}$
and converging to $\Lambda_1$ as $t \to +\infty$ for more general Hamiltonian systems which
cannot be resolved explicitly. For instance, assuming that $X$ is a closed manifold,
we find such a trajectory for every Hamiltonian $H(p,q,z)$
which equals $-cz$ near $\Lambda_1$ and is positive and bounded on $\{z < 0\}$ ,
see Examples~\ref{exam-1-jet-space-0-section-and-its-Reeb-shift-homol-bonded} and \ref{exam-infty} below.}
\end{exam}

We refer to \cite{MS,WWY} for the existence and fine structure of global attractors of
flows generated by certain contact Hamiltonians convex in the momenta variables $p$.
These results were established by methods of the calculus of variations in the spirit of
the Aubry-Mather theory.

\subsection{Non-equilibrium thermodynamics}

The second theme of the present paper is application of the results outlined above
to non-equilibrium thermodynamics. In the geometric model of equilibrium thermodynamics (see e.g. \cite{HermannR-book73}) the thermodynamic phase space is described as the standard contact space $(\R^{2n+1}, \xi)$, $n\geq 1$, $\xi:= \Ker \lambda$, $\lambda = dz- pdq$, $p=(p_1,\ldots,p_n)$, $q=(q_1,\ldots,q_n)$.
Here  $z$, up to a sign, is a thermodynamic potential, such as internal energy or free energy, and $(p_i, q_i)$ are pairwise conjugate variables, such as {\it (temperature, entropy)} or
{\it (magnetization, external magnetic field)}. The fundamental thermodynamic equilibrium relation reads
\[
dz= \sum_{i=1}^n  p_i dq_i.
\]
Consequently, the set of equilibrium states of a thermodynamic system forms a properly embedded Legendrian submanifold $\Lambda_1 \subset (\R^{2n+1},\xi)$.


When a system in an equilibrium state undergoes a perturbation, it
either moves along the submanifold of equilibrium states
or  moves to a non-equilibrium state and then enters a dynamical process. {\it Relaxation processes}, i.e. a gradual return to the equilibrium, are of special interest. For instance, Prigogine  in his 1977 Nobel lecture \cite{Pri} emphasizes
the local stability of the thermodynamic equilibrium due to the fact that thermodynamic
potentials serve as Lyapunov functions near the equilibrium, and addresses a question
"Can we extrapolate this stability property further away from equilibrium?" (p. 269).
A number of papers \cite{Mrugala-et-al-RepMathPhys91,Mrugala-RepMathPhys2000, Bravetti,Haslach,Grmela-Ottinger-PRE97, Grmela-PhysA2002, Goto-JMP2015,V-Maschke,Leon}
propose to describe relaxation processes of non-equilibrium thermodynamics in terms of contact dynamics. The ways of choosing a contact Hamiltonian vary slightly from paper to paper -- the equation of motion of one of the parameters is provided by physical considerations, and the contact Hamiltonian is chosen to yield this equation (cf. a detailed comparison between
the standard physical and the contact models for the Ising model in Section~\ref{subsec-ising-1}). After a brief warm up with Newton's law of cooling illustrating
the main notions of contact thermodynamics, we pass to a detailed study of the Glauber dynamics of the Ising model in the mean field approximation, aiming at detecting relaxation processes. Let us state here an informal version of our main result in this direction -- see Theorem~\ref{thm-ising} below for a rigorous formulation.

\bigskip
\noindent
{\bf ``Theorem" (an informal version of Theorem~\ref{thm-ising} below).}
{\it If an equilibrium of the Ising model is disturbed (due to a sudden change of parameters of the model), then, under certain constraints on the parameters, there exists a relaxation process whose initial conditions lie in the perturbed equilibrium that asymptotically converges to the original equilibrium.}
\bigskip

The proof is based on the above-mentioned existence result for semi-infinite trajectories
for a pair of Legendrians. Furthermore, we provide a microscopic level interpretation for contact Hamiltonians governing the relaxation, see Section \ref{subsec-pG}.

\subsection{{Organization of the paper}}
The rest of the paper is organized as follows. In Section~\ref{sec-global} we introduce interlinking, a
 notion from contact dynamics which is crucial for our purposes, and present the main results
on existence of semi-infinite trajectories connecting a {robustly interlinked pair of} Legendrians. At the end of this section
we review the existence and structural stability of Legendrian global attractors $\Lambda$ of the contact Hamiltonian flows of contact Hamiltonians $H$ vanishing on $\Lambda$. (Such pairs $H,\Lambda$ play a central role in contact non-equilibrium thermodynamics).

Section~\ref{sec-therm} deals with applications to thermodynamics.

\section{Stability and interlinking}\label{sec-global}

\subsection{Attractors and repellers} \label{subsec-attractors-repellers}
We start with general preliminaries on attractors and repellers. Given a complete continuous flow $\{ \varphi^t\}$, $t\in\R$, on a topological space $Z$, a set $Y\subset Z$ and a compact subset $Q\subset \text{Closure}(Y)$, we say that $Q$ is a {\it local attractor in $Y$ of the flow $\{ \varphi^t\}$} if it is invariant under the flow and there exists a neighbourhood
$U$ of $Q$ with the following property: for every $x \in U\cap Y$ and for every neighbourhood $U'$
of $Q$ there exists $t_0$ such that $\varphi^t (x) \in U'$ for all $t>t_0$. Such a neighbourhood $U$ will be called an {\it attracting neighbourhood of $Q$ for $Y$ and $\{ \varphi^t\}$}. If the same holds
for all $t < -t_0$, the set $Q$ is called a {\it local repeller in $Y$ of the flow $\{\varphi^t\}$} and the corresponding neighbourhood $U$ is called a {\it repelling neighbourhood of $Q$ for $Y$ and $\{ \varphi^t\}$}. If an attracting/repelling neighbourhood $U$ can be chosen so that $Y\subset U$, we say that $Q$ is a {\it global attractor/repeller} in $Y$ of the flow $\{ \varphi^t\}$. A global attractor/repeller in $Z$ is simply called {\it attractor/repeller}.

\subsection{Interlinking} \label{subsec-interlink}
In what follows we work on a possibly non-compact contact manifold $\Sigma$ equipped with
a contact form $\lambda$ whose Reeb flow is complete. All contact Hamiltonians are assumed to be time-independent and to have a complete contact Hamiltonian flow (such contact Hamiltonians are called complete).

An {\it ordered} pair {$(\Lambda_0,\Lambda_1)$ of disjoint closed (i.e., compact without boundary) Legendrian submanifolds of $(\Sigma,\xi)$} is called {\it interlinked} (cf. \cite{EP-tetragons}) if there exists a constant $\mu= \mu(\Lambda_0,\Lambda_1,\lambda) >0$ such that every (complete) bounded strictly positive contact Hamiltonian $F$ on $\Sigma$ with $F \geq a >0$ possesses an orbit  of time-length $\leq \mu/a$ starting at $\Lambda_0$ and arriving at $\Lambda_1$.

\begin{rem}
{\rm
It is easy to see that if the pair $(\Lambda_0,\Lambda_1)$ is interlinked for some contact form $\lambda$ on $\Sigma$ whose Reeb flow is complete, then it is interlinked for any other such contact form $\lambda'$ defining the same coorientation of the contact structure as $\lambda$, provided the ratio $\lambda/\lambda'$ is bounded away from $0$ and $+\infty$.  At the same time, the constant $\mu(\Lambda_0,\Lambda_1,\lambda)$ does depend on $\lambda$.
}
\end{rem}
\bigskip

The interlinking is called {\it robust} if the same holds true, with possibly a different constant $\mu'$, for every pair of Legendrians $\Lambda'_0$ and $\Lambda'_1$ from sufficiently small $C^\infty$-neighbourhoods of $\Lambda_0$ and $\Lambda_1$, respectively. Existence of (robustly) interlinked pairs is a non-trivial
phenomenon detected by methods of ``hard" contact topology. {We refer to \cite{EP-Leg-pers-mod, EP-big} for a detailed discussion, as well as to \cite{EP-tetragons} for a related notion in symplectic Hamiltonian dynamics.}

\medskip

\begin{exam} [\cite{EP-Leg-pers-mod}, Theorem 1.5(i)]
\label{exam-1-jet-space-0-section-and-its-Reeb-shift-homol-bonded}
{\rm Let $\Sigma$ be the jet space $J^1 X = T^* X\times \R (z)$ of a closed manifold $X$
equipped with the standard contact form ( see Example~\ref{exam-onejet} above).  Let $\psi$ be a negative function on $X$, and let $\Lambda_0:= \{z=\psi(q),p=\psi'(q)\}$ be the graph of its 1-jet. Let $\Lambda_1$ be the zero section. Then the pair $(\Lambda_0,\Lambda_1)$ is interlinked.
}
\end{exam}

\begin{exam} [\cite{EP-Leg-pers-mod}, Theorem 1.5(ii)]
\label{exam-onechord}
{\rm Let $\Sigma = J^1X$ be as in the previous example. Let $\Lambda_0 \subset J^1 X$ be a  Legendrian submanifold with the following properties. First, $\Lambda_0$ is isotopic to the zero section $\Lambda_1$ through Legendrian submanifolds,
and second, there is a unique chord of the Reeb flow $R^t$ starting on $\Lambda_0$ and ending on  $\Lambda_1$. Denote the chord by $R^t x$, $t \in [0,\tau]$ with $x \in \Lambda_0$ and $y:= R^\tau x \in \Lambda_1$
Assume the following {\it non-degeneracy condition}:
\begin{equation}\label{eq-nondeg}
D_x R^\tau(T_x\Lambda_0) \oplus T_y \Lambda_1 = \xi_y\;,
\end{equation}
where $\xi_y$ stands for the contact hyperplane at $y$.
Then the pair $(\Lambda_0,\Lambda_1)$ is interlinked.
}
\end{exam}

\subsection{Asymptotic trajectories-1}\label{subsec-as-1}
Let $H$ be a (not necessarily bounded) contact Hamiltonian on a contact manifold $\Sigma$ equipped with a contact form $\lambda$ whose Reeb flow is complete.
Let a submanifold $M \subset \Sigma$ be the union of a finite number of the connected components of the
nodal set $\{H=0\}$ so that $M$ separates $\Sigma$ into two open parts, $\Sigma_-$ and $\Sigma_+$.
Assume that $H$ is strictly positive on $\Sigma_+$, but in general it is allowed to change sign on $\Sigma_-$. Recall that $R$ stands for the Reeb vector field of the contact form $\lambda$. Let $\Lambda_1 \subset M$  be a closed Legendrian submanifold. Let us make the following assumption.

\medskip
\noindent
{\bf Assumption $\clubsuit$:}
\begin{itemize}
\item [{(i)}] There exists $\kappa_1 > 0$ such that $dH(R) \leq 0$
on $\{0 < H <\kappa_1 \} \cap \Sigma_+$.
\item [{(ii)}] There exists $\kappa_2 > \kappa_1 > 0$ such that $dH(R) \leq 0$ on $\{H \geq \kappa_2\} \cap \Sigma_+$.
\item [{(iii)}] $dH(R)<0$ near $\Lambda_1$.
\item[{(iv)}] $\Lambda_1$ is a local attractor in $\Sigma_+ \cup \Lambda_1$ of the contact flow generated by $H$.
\end{itemize}

\medskip
\noindent
\begin{exam}\label{exam-infty} {\rm Consider $\Sigma= T^*X\times \R$ with the standard
contact form $dz-pdq$, where $X$ is a closed manifold.  Let $H: \Sigma \to \R$ be a function which equals $-cz$, $c >0$, near the zero section $\Lambda_1:= \{p=z=0\}$, and is bounded and positive on $\Sigma_+ := \{z <0\}$. {Let $M:=\{ z=0\}$.} Then assumption $\clubsuit$ holds.}
\end{exam}

\medskip
\noindent
\begin{rem}{\rm {Suppose that the contact vector field $v$ of $H$, when restricted to $M$, admits a local Lyapunov function near $\Lambda_1$},
i.e., there exists a non-negative smooth function $G$ on a neighbourhood $Z$ of $\Lambda_1$ in $M$ which vanishes on $\Lambda_1$ and which satisfies $dG(v) <-0$ on $Z \setminus \Lambda_1$. Together with assumption $\clubsuit$(iii) this
yields $\clubsuit$(iv). The proof is analogous to the one of Proposition~\ref{prop-local-atrep} below.
}
\end{rem}

\begin{thm}\label{thm-A}
Let $H: \Sigma \to \R$ be a contact Hamiltonian,
and let
$$\Lambda_1 \subset M \subset \{H=0\}$$
be a Legendrian submanifold, as above. Assume that
$\Sigma, H, M, \Lambda_1$ satisfy assumption
$\clubsuit$.
Let  $\Lambda_0$ be another closed Legendrian submanifold lying in the set $\{H \leq \kappa_2\}$, where $\kappa_2$ is the constant in $\clubsuit$(ii).
Assume that $(\Lambda_0, \Lambda_1)$  is robustly interlinked pair.
Then $\Lambda_0 \cap \Sigma_+ \neq \emptyset$, and there exists a trajectory of the contact Hamiltonian flow of $H$, fully lying in $\Sigma_+$, which starts on $\Lambda_0$ and asymptotically converges to $\Lambda_1$.
\end{thm}

\begin{proof}
Write $v$ for the contact vector field of $H$ and write $U$ for the local
attracting neighbourhood of $\Lambda_1$ in $\Sigma$.

Let $\Lambda'_1$ be the image of $\Lambda_1$ under the time-$\epsilon$ Reeb
flow, with $\epsilon < 0$ and $|\epsilon|$ small enough. Then, by $\clubsuit$ (iii), we have that  $\Lambda'_1 \subset U \cap \Sigma_+$,
and  the pair $(\Lambda_0,\Lambda'_1)$ is interlinked. Furthermore,
$H > b$ on $\Lambda'_1$ with some $0 <b <\kappa_1$,
where $\kappa_1$ is a constant from $\clubsuit$ (i).

It suffices to find a trajectory of the flow of $v$ connecting $\Lambda_0$ and $\Lambda'_1$ in $\Sigma_+$. Indeed, since $U$ is attracting, such a trajectory will necessarily  converge to $\Lambda_1$.

 To this end, consider a non-decreasing smooth function $u:\R\to\R$
with $u(s)=s$ for $s \in [b/2,3\kappa_2/2]$, $u(s)=b/4$ for $s \leq b/4$, and $u(s)=2\kappa_2$ for $s \geq 2\kappa_2$.
Define a new contact Hamiltonian $H'$ on $\Sigma$ by $H'= u(H)$ on $\Sigma_+$ and
$H'=b/4$ on $\Sigma_-$. Write $v'$ for the contact vector field of $H'$.

 By interlinking, there exists
a trajectory $\gamma$ of $v'$ connecting $\Lambda_0$ and $\Lambda'_1$.
Since $$dH'(v') = H'dH'(R) = H'\left(\frac{du}{ds}\circ H\right)dH(R) \leq 0$$
on $\{b/2 \leq H \leq b\} \cap \Sigma_+$, every trajectory of $v'$ starting in
$\{H = b/2\} \cap \Sigma_+$ cannot enter the domain $\{H \geq b\}$, and hence does not reach $\Lambda'_1$.  Similarly, no trajectory of $v'$ starting in $\{H \leq \kappa_2\}\cap \Sigma_+$ can exit to $\{H \geq 3\kappa_2/2\} \cap \Sigma_+$. It follows that
$\gamma$ lies fully in the set $\{b/2 \leq H \leq 3\kappa_2/2\} \cap \Sigma_+$, and hence is a trajectory of $H$. This proves all the statements of the theorem.
\end{proof}

\subsection{Asymptotic trajectories-2}
Let $H:\Sigma\to\R$ be a contact Hamiltonian and let $v$ be its Hamiltonian vector field.
Let $W$ be a connected component
of $\{H >0\}$ with smooth compact boundary $M$.
Recall that $R$ stands for the Reeb vector field.

Observe that $M$,  being the union of some connected components  of $\{H=0\}$, is invariant under the Hamiltonian flow $\{ \varphi^t\}$ of $H$. In the next theorem, we allow $dH(R)$ to change sign along $M$.
To this end we need the notion of the
{\it separating hypersurface} defined by
\[
\Gamma := \{\ x \in M \;:\; d_xH(R) = 0\ \}.
\]
Set
\[
M_-:= \{\ x \in M \;:\; d_xH(R) < 0\ \},\;\; M_+:= \{\ x \in M \;:\; d_xH(R) > 0\ \}.
\]

\medskip\noindent
{\bf Assumption $\spadesuit$:} $\Gamma$ is a smooth hypersurface; at the points of
$\Gamma$ the contact vector field $v$ is transversal to $\Gamma$ and points into $M_-$.

\medskip\noindent
\begin{question} Do there actually exist examples where $\Gamma$ is a smooth hypersurface and at the points of
$\Gamma$ the contact vector field $v$ is transversal to $\Gamma$ and points into $M_+$?
\end{question}

\begin{rem} {\rm
It would be interesting to compare assumption $\spadesuit$ to the notion of convexity
of compact hypersurfaces \cite{Giroux}. For instance, by Proposition 2.1 in \cite{Giroux},
$M$ is convex\footnote{We thank E.Giroux for bringing our attention to
convexity in this context.} when $\dim \Sigma =3$, $\dim M=2$.
}
\end{rem}

Set
\[
U^-_s := \varphi^s(M_-),\ s >0,
\]
\[
U^+_s := \varphi^s(M_+),\ s < 0,
\]
\[
Q_- := \bigcap_{s >0} U^-_s,
\]
\[
Q_+ := \bigcap_{s <0} U^+_s.
\]
Note that $Q_-\subset M_-$, $Q_+\subset M_+$.
We call the sets
$Q_-$, $Q_+$ the {\it cores} of $M_-$ and of $M_+$, respectively.

We claim that if assumption $\spadesuit$ holds, then $Q_-$ is {\it a global} attractor of the flow $\{ \varphi^s\}$ on $M \setminus Q_+$.

Indeed, the flow trajectory of any point in $M_-$ asymptotically converges to $Q_-$ (by the definition of $Q_-$ and compactness of $M$).
Also, the flow trajectory of any point in $M_+ \setminus Q_+$, by the definition of $Q_+$, arrives at a small neighbourhood of $\Gamma$ in $M_+$. Assumption $\spadesuit$ guarantees that it then crosses $\Gamma$ into $M_-$,
and hence, again by the definition of $Q_-$ and the compactness of $M$, asymptotically converges to $Q_-$. This proves the claim.

\begin{thm}
\label{thm-semimain} With $H$, $W$ and $M$ as above, suppose that assumption $\spadesuit$ holds.
Assume furthermore that $Q_- := \Lambda$ is a closed Legendrian submanifold.

Then  for every closed Legendrian $K \subset W$,  such that the pair $(K,\Lambda)$ is robustly interlinked, there exists a semi-infinite trajectory of the contact Hamiltonian flow $\{ \varphi^t\}$ of $H$ starting on $K$ and asymptotically converging to $\Lambda$.
\end{thm}

\begin{rem}
\label{rem-attractors-exist-hyperbolicity}
{\rm
The situations where a global attractor of a contact flow is a closed Legendrian submanifold do exist -- in Section~\ref{sec-loc} we show how the existence and structural stability of such attractors can be proved using the theory of hyperbolic dynamical systems.
}
\end{rem}
\bigskip

As a preparation for the proof of Theorem~\ref{thm-semimain}, we need the following result.

\begin{prop}\label{prop-local-atrep}
 The cores $Q_-$ and $Q_+$ are a local attractor and a local repeller of the flow
$\{\varphi^t\}$ on $\Sigma$.
\end{prop}

Let us emphasize that the point of the proposition is that the property of being a local attractor/repeller
holds {\sl in the ambient manifold $\Sigma$}.

\bigskip
\noindent
{\bf Proof of Proposition~\ref{prop-local-atrep}:} We prove this for the negative core $Q=Q_-$.
The core $Q$ is compact, and hence we can assume that $dH(R) \leq -c <0$ on $Q$ for some $c>0$.
In particular, $dH \neq 0$ in a neighbourhood of $Q$. Fix an auxiliary Riemannian metric
on $\Sigma$. \color{black} Write $\nabla$ for the gradient with respect to this metric,
and $|\cdot|$ for the Riemannian length of a tangent vector. \color{black}
Using
the flow of $\nabla H/|\nabla H|^2$, we identify a neighbourhood of $Q$ with $V:= U^-_{s_0} \times (-\epsilon_0,\epsilon_0)$, for some $\epsilon_0>0$ and a sufficiently large $s_0>0$. Here the Hamiltonian $H$ becomes the projection to the second factor.

Note that
\begin{equation}\label{eq-decay}
dH(v) = H dH(R)\;.
\end{equation}
There exists a sufficiently large $s >s_0$ and a sufficiently small $\epsilon \in (0,\epsilon_0)$ so that $v$ is transversal to the boundary of $\Pi \subset V$ for $\Pi:= U^-_s \times (-\epsilon,\epsilon)$ and points inside the domain $\Pi$
-- because of \eqref{eq-decay} and because the boundary of $U^-_s$ is the image of $\Gamma$ under $\varphi^s$ and $v$ is transversal to $\Gamma$ and points into $M_-$ (by assumption $\spadesuit$).
Thus $\Pi$ is an ``isolating block": every trajectory of $\{\varphi^t\}$ starting in $\Pi$ remains inside $\Pi$ as $t >0$. Hence, by \eqref{eq-decay} we get that for any $x_0\in \Pi$
\begin{equation} \label{eq-decay-1}
H(\varphi^t(x_0)) \to 0 \;\;\text{as}\;\; t \to +\infty\;.
\end{equation}
Take any $x_0 \in \Pi$ and consider the set $A$ of points $x \in \Pi$ such that $x = \lim \varphi^{t_i}(x_0)$ for some sequence of times $t_i \to +\infty$.
We claim that $A \subset Q$.
Indeed, take any $x \in A$. By \eqref{eq-decay-1}, we conclude that $H(x)=0$, so that $x \in U^-_s$.
Assume, by contradiction, that $x \notin Q$. Then we can choose a smaller isolated block $\Pi'$ whose
closure is contained in $\Pi$, and such that $x \notin \Pi'$. Since $Q$ is an attractor of $\{\varphi^t\}$ in $U^-_s$, there exists a neighbourhood $Z$ of $x$ disjoint from $\Pi'$,  and $T >0$ such that $\varphi^T(Z) \subset \Pi'$. Since $\Pi'$ is an isolated block, $\varphi^t(Z) \subset \Pi'$ for any $t \geq T$. It follows that if $\varphi^t (x_0) \in Z$ for some $t\geq T$, then $\varphi^r (x_0) \notin Z$ for all $r > t+T$, and hence $x \notin A$. We get a contradiction.

Therefore $A\subset Q$. This implies that the trajectory $\varphi^t (x_0)$ asymptotically converges to $Q$ -- indeed, otherwise, by the compactness of $\text{Closure}(\Pi)$, we would get a sequence $t_i \to +\infty$ such that the sequence $\{ \varphi^{t_i}(x_0)\}$ has a limit that does not lie in $Q$, in contradiction to $A\subset Q$.

This shows that $Q$ is a local attractor in $\Sigma$ of the flow $\{\varphi^t\}$.
\Qed

\bigskip\noindent
{\bf Proof of Theorem~\ref{thm-semimain}: }
Perturb $\Lambda$ inside the nodal domain $W$ of $H$ by the (reversed) Reeb flow.
Denote the obtained Legendrian by $\Lambda'$. We can assume that $\Lambda'$ is contained in
the attracting neighbourhood of $\Lambda$ built in Proposition~\ref{prop-local-atrep}.

 Choose $\delta >0$ small enough such that
$\Lambda' \subset \{H > \delta\}$. Modify $H$ outside $\{H \geq \delta\}$ to a positive Hamiltonian $H_\delta$ bounded away from zero. Denote by $M_\delta$ the boundary of $\{H > \delta\} \cap W$.

By interlinking, there exists a trajectory of the contact Hamiltonian flow of $H_\delta$ connecting $K$ with $\Lambda'$.
If it contained in $\{H > \delta\}$ we are done: this trajectory is a trajectory of the contact Hamiltonian flow of $H$, and it asymptotically converges
to $\Lambda$ as it passes through the attracting neighbourhood. If not, it must exit the
domain $\{H > \delta\}$ through $M_\delta$. Denote the obtained piece of trajectory lying
fully in $\{H \geq \delta\}$ by $\gamma_\delta$. Passing to a subsequence of $\{\gamma_\delta\}$ as $\delta \to 0$, and remembering that $M$ is compact and invariant under
$\{ \varphi^t\}$, we get a semi-infinite trajectory $\gamma:= \{\varphi^t x_0\}_{t\geq 0}$, $x_0 \in K$,  which
has limit points in $M$. The set $A \subset M$ of these limit points is invariant under $\{\varphi^t\}$. Note that, by Proposition~\ref{prop-local-atrep}, the core $Q_+$  is a local repeller in $\Sigma$. Thus,
$A \subset M \setminus Q_+$.
But since $\Lambda$ is the global attractor of $\{\varphi^t\}$ in  $M \setminus Q_+$,
the closure of $A$ intersects $\Lambda$. Therefore, our trajectory $\gamma$ must enter the attracting neighbourhood of $\Lambda$. Hence it converges to $\Lambda$,
and the proof is complete.
\qed

\begin{exam}[{\bf Contact M\"{o}bius dynamics}]
\label{exam-contact-mobius-dynamics}
{\rm
Consider the contact manifold $J^1 \SP^1 = \R^2(z,p) \times \SP^1(q)$ equipped
with the contact form $\lambda := dz-pdq$. We start our discussion with the analysis
of the contact Hamiltonian
\[
F(p,q,z) = z^2+p^2 -1.
\]
The contact vector field generated by $F$ is
\[
v=2pz \frac{\partial}{\partial p} - 2p \frac{\partial}{\partial q} + (z^2 - p^2 - 1) \frac{\partial}{\partial z}.
\]
The corresponding contact Hamiltonian system is integrable, and its dynamics admits a
simple geometric description. We equip $\R^2(z,p)$ with a complex coordinate $w= z+ip$, identifying it with $\C$. Replace the contact form $\lambda$ by $p^{-1}\lambda = -dq+ p^{-1}dz$.
This form has a singularity at the real line $\ell:= \{p=0\}$. The line $\ell$ splits $\C$
into the upper half-plane $\mathbb{H}^+$ and the lower half plane $\mathbb{H}^-$.
Both half-planes are equipped with the hyperbolic area form $\Omega:= p^{-2}dz \wedge dp$.
 With this language
the contact manifold $J^1 \SP^1\setminus \{ p=0\}$ can be considered as the prequantization of $(\C\setminus\ell, \Omega)$.

The group $PSL(2,\R)$ acts on $\C$ by hyperbolic isometries of both half-planes. In particular, it preserves the singular line $\ell$ and the form $\Omega$. One can check that the action lifts, in the standard manner, to the prequantization $J^1 \SP^1\setminus \{ p=0\}$. (The latter
statement requires a calculation: one has to verify that for every $g \in PSL(2,\R)$
the form $g^*(p^{-1}dz) - p^{-1}dz$ is regular everywhere except for the point $g^{-1}(\infty)$.
We leave the verification to the reader).

With this preliminaries, the contact flow of the contact Hamiltonian $F$
coincides with the contact lift of a one-parameter subgroup of M\"{o}bius transformations
having an unstable fixed point at $1$ and a stable fixed point at $-1$, see Figure 1.

\begin{figure}[h]
\label{fig-1}
\caption{}
\includegraphics[width=0.5\textwidth]{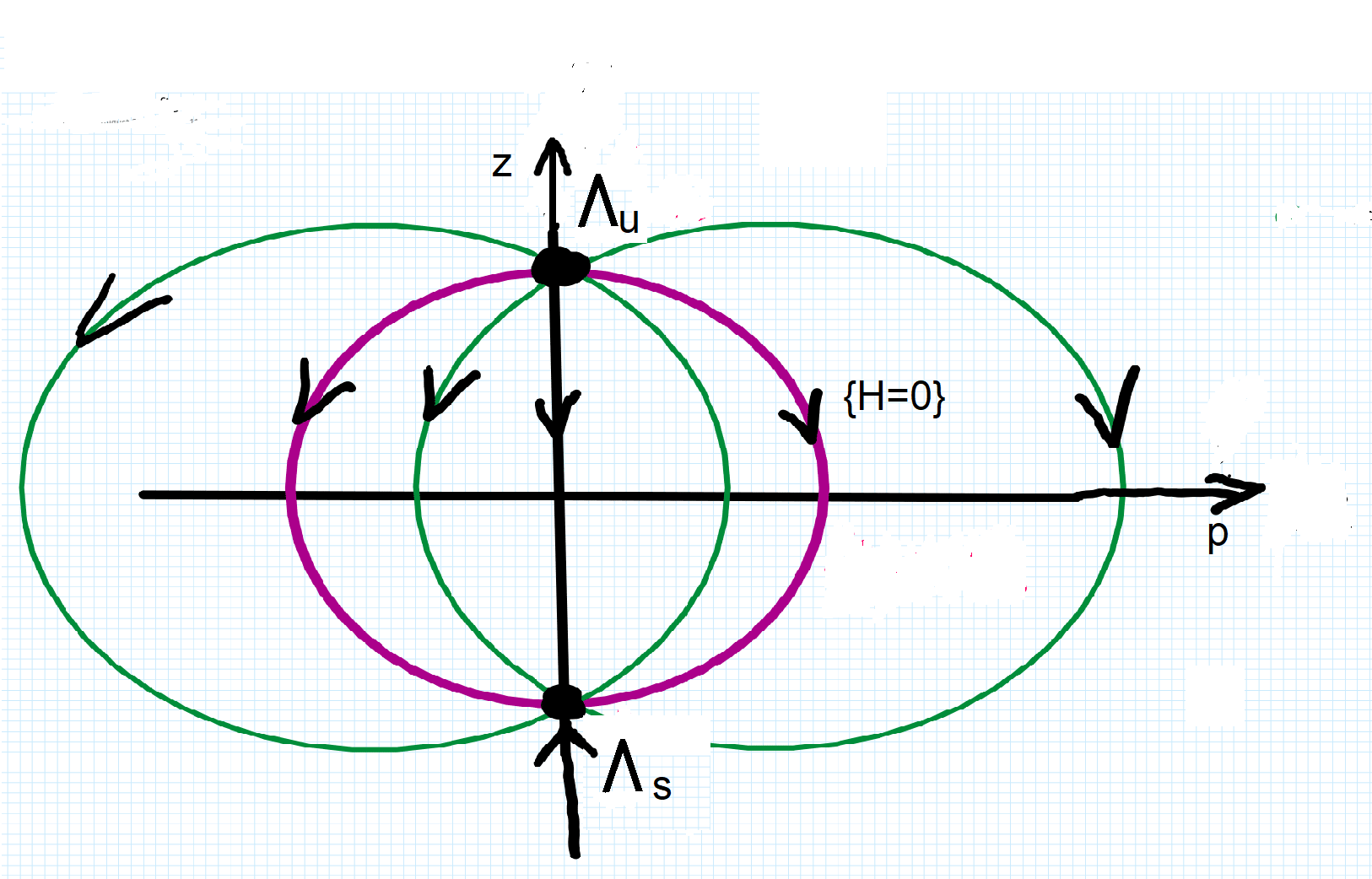}
\centering
\end{figure}

A calculation (that we leave to the reader) yields:
\[
w(t) = \frac{ w(0)\cosh t - \sinh t}{-w(0)\sinh t + \cosh t}.
\]

However, there is a problem: the latter trajectories may escape
to infinity in finite time, and therefore the contact Hamiltonian flow of $F$ is incomplete.

We shall remedy this by introducing an appropriate cut-off and replacing $F$ with a new complete Hamiltonian $H$. Fix a smooth function
$a :[0,+\infty) \to [-1,+\infty)$ so that $a'(s)>0$ for all $s$,
$\lim_{s \to +\infty}a(s)= a_\infty >1$ and $a(s)=s-1$ for all
$s \in [0,1+\epsilon]$ for some $\epsilon >0$. One readily checks that the contact Hamiltonian
flow of $H(p,q,z) := {a} (p^2+z^2)$ is complete, and the dynamics of $H$ coincides with the
one of $F$ (described above) on the solid torus $\{p^2+z^2 < 1+\epsilon\}$, see Figure 2. Moreover, $H$ is bounded from above.

\begin{figure}[h]
\label{fig-2}
\caption{}
\includegraphics[width=0.5\textwidth]{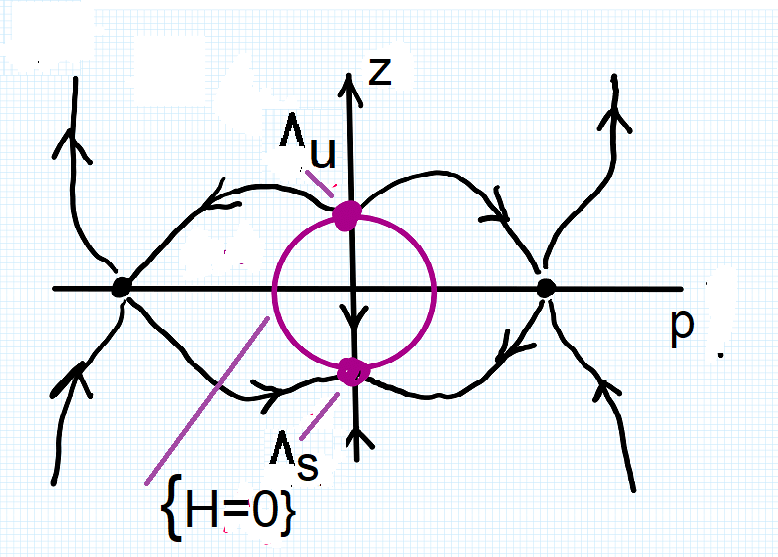}
\centering
\end{figure}

Let us observe that the 2-torus $\{ F = 0\} = \{ H = 0\}$ is convex in the sense of contact topology \cite{Giroux}: it is transverse to the contact vector field $z \partial/\partial z +
p\partial/\partial p$. This field is tangent to the contact structure along
the circles $\{p={\pm 1,z= 0}\}$ which split the torus into two annuli. Each of these annuli
is Liouville with respect to the symplectic structure $dz \wedge dp$, with the Liouville
vector field being parallel to the characteristic foliation. The Legendrian curves
$\Lambda_{\rm s} := \{ {p=0,z=-1}\}$ and $\Lambda_{\rm u} := \{{p=0, z= 1}\}$ are Lagrangian skeleta
of the annuli (here ${\rm s}$ and ${\rm u}$ stand for stable and unstable, respectively).

Note that $\Lambda_{\rm s}$ is a local attractor of the contact Hamiltonian flow of $H$. In fact, all the trajectories of the flow, except those departing from $\Lambda_{\rm u}$ in the direction of the $z$-axis,
asymptotically converge to $\Lambda_{\rm s}$.

Now we are ready to illustrate the statement of Theorem~\ref{thm-semimain}.
The Legendrian submanifolds $K_c:= \{p=0, z= c\}$ lie in $W= \{p^2+z^2 > 1\}$ whenever
$|c| >1$. By an obvious modification of Example~\ref{exam-1-jet-space-0-section-and-its-Reeb-shift-homol-bonded},
the pair of Legendrian submanifolds $(K_c, \Lambda_{\rm s})$ is interlinked
if $c <-1$, and is not interlinked if $c >1$.
In the former case, $K_c$ lies in the attracting neighbourhood of $\Lambda_{\rm s}$ for the contact Hamiltonian flow of $H$, and in the latter
case no trajectory of the flow starting on $K_c$ converges asymptotically to $\Lambda_{\rm s}$.
}
\end{exam}

\begin{rem} {\rm
The form $\Omega:= p^{-2}dz \wedge dp$ can be considered as a {\it $b$-symplectic form} in the sense of Miranda and Oms \cite{MO}. An intriguing dynamical object arising in $b$-contact topology is {\it an escape orbit} -- that is, an infinite Reeb orbit converging to
the singularity, see a recent paper \cite{MOP} by Miranda, Oms and Peralta-Salas. It would be interesting to explore a possible link between escape orbits and the asymptotic orbits appearing in our story.}
\end{rem}

\subsection{Local and structural stability}
\label{sec-loc}
Here we discuss the existence and structural stability of Legendrian global attractors of contact flows mentioned in Remark~\ref{rem-attractors-exist-hyperbolicity} above.

The hyperbolic theory of dynamical systems enables one to detect existence of attractors
in terms of local data. Let us discuss this briefly in the context of contact Hamiltonian
flows. Let $(\Sigma^{2n+1},\xi)$ be a contact manifold with a contact form $\lambda$ and
the corresponding Reeb vector field $R$. Let $H: \Sigma \to \R$ be a contact Hamiltonian
having $0$ as a regular value. Recall that the hypersurface $M:= \{H=0\} \subset \Sigma$ is invariant under the contact Hamiltonian flow $\{\varphi^t\}$ of $H$. Let $\Lambda \subset M$
be a smooth closed $n$-dimensional submanifold invariant under $\{\varphi^t\}$.

Let us introduce an auxiliary Riemannian metric on $M$ near $\Lambda$. The tangent bundle to $M$  along $\Lambda$ splits into the direct sum  $T_\Lambda M = T\Lambda \oplus N\Lambda$ of the tangent and normal bundles to $\Lambda$. Denote by $\pi: T_\Lambda M \to N\Lambda$ and $\tau: T_\Lambda M \to T\Lambda$  the corresponding projections.

\bigskip

We say that $\Lambda$ is $(a,b)$-{\it normally hyperbolic} invariant subset of $\{\varphi^t\}$ on $M$ with $a > |b| >0$, if $||\pi D_x\varphi^t||  \leq Ce^{-at}$ and $||\tau D_x\varphi^t||  \leq Ce^{bt}$ for some $C >0$ and all $x \in \Lambda$. \color{black} Here $||\cdot||$ stands for the
operator norm with respect to the Riemannian metric. \color{black}
By Theorem 4 of \cite{Fenichel}, such an invariant submanifold is a local attractor of $\{\varphi^t\}$ on $M$.

\bigskip

Further, if $\Lambda$ is $(a,b)$-{\it normally hyperbolic} for every $b>0$, then
$\Lambda$ is $C^\infty$-{\it structurally stable}. This means that for every $\epsilon > 0$ there exists a $C^2$-neigbourhood $\cU$ of $H$ (in the strong Whitney topology)
such that for every $H' \in \cU$ the hypersurface $M' = \{H'=0\}$ contains an $(a',b')$-normally hyperbolic invariant submanifold $\Lambda'$ smoothly diffeomorphic to $\Lambda$ with $a' > a-\epsilon$ and $|b'| < \epsilon$. This follows from Theorems 1 and 2
in \cite{Fenichel}. In fact, the proof in \cite{Fenichel} shows that $\Lambda'$ is $C^\infty$-close to $\Lambda$. This yields, by \cite{Fenichel},
that $\Lambda'$ itself is $C^r$-structurally stable with some large, albeit finite, $r$.

\bigskip
{\bf Warning:} $\Lambda'$ is {\bf not} normally hyperbolic
with every $b >0$, but with some positive $b' <\epsilon$.

\bigskip

Assume that $\Lambda$ is $(a,b)$-normally hyperbolic invariant subset of $\{\varphi^t\}$ on $M$. Impose an extra condition on $H$ at $\Lambda$: for some $c>0$,
\begin{equation}
\label{eq-H-R}
d_xH(R) \leq -c <0 \;\; \forall x \in \Lambda.
\end{equation}

Observe that this condition guarantees that $R$ is transversal to $M$ near $\Lambda$.
We shall assume without loss of generaility that our auxiliary Riemannian metric is extended to a neighbourhood of $\Lambda$ in the whole $\Sigma$,  and that $R$ is orthogonal to $M$ near $\Lambda$.
Furthermore, writing $v$ for the vector field of $\{\varphi^t\}$ on $\Sigma$, we have
\begin{equation}\label{eq-lambda-H-R}
L_v\lambda = dH(R)\lambda,
\end{equation}
and hence one readily concludes that assumption \eqref{eq-H-R}
yields $(d,b)$-normal hyperbolicity of $\Lambda$ in $\Sigma$ with $d= \min(a,c)$,
provided $|b| < c$.

\begin{thm} If $|b| < c$, the submanifold $\Lambda$ is Legendrian.
\end{thm}
\begin{proof} The proof resembles that one of the fact the center manifold of
a zero of the Liouville vector field is isotropic, see \cite[Prop. 11.9]{CiEl}.
Take a point $x \in \Lambda$ and a unit tangent vector $\eta \in T_x\Lambda$.
By \eqref{eq-lambda-H-R} and the definition of $b$, we have that for some $C' >0$ and all $t \geq 0$,
$$|\lambda(\eta)| = \left|\left((\varphi^t)^*\lambda\right)\left((\varphi^t_*)^{-1} \eta \right)\right| \leq C'  e^{-ct} e^{bt}.$$
Thus, $\lambda(\eta)=0$, because the right hand side converges to $0$ as $t \to +\infty$.
\end{proof}

We sum up our discussion as follows.
\begin{thm}
Let $\Lambda$ be a closed Legendrian submanifold lying in a regular level set $M=\{H=0\}$ of a (complete) contact Hamiltonian $H:\Sigma\to\R$. Assume that
$dH(R) \leq -c$ along $\Lambda$ for some $c>0$, and that $\Lambda$ is $(a,b)$-normally hyperbolic for every $b>0$.

Then
\begin{itemize}
\item[{(i)}] $\Lambda$ is a local attractor of the contact Hamiltonian flow of $H$;
 \item[{(ii)}] For every $C^2$-small  perturbation $H'$ of $H$, the level $\{H'=0\}$
contains a Legendrian submanifold $\Lambda'$ which is $C^\infty$-close to $\Lambda$ and
which is attracting for the contact  Hamiltonian flow of $H'$.
\end{itemize}
\Qed
\end{thm}

\begin{exam}
{\rm
Let $\Sigma := J^1(X)$, $\lambda :=dz-pdq$. Consider a contact Hamiltonian $H(p,q,z) = -a(z - \phi(q))$, $a >0$. Such Hamiltonians will appear in our discussion on thermodynamics -- see formulas \eqref{eq-Ham-isentropic} and \eqref{eq-contHam-b} below.
Consider the Legendrian submanifold
\begin{equation}\label{eq-therm-lambda}
\Lambda := \{\ z=\phi(q),\ p=\phi'(q)\ \}.
\end{equation}
Then $dH(R) = -a$, $\Lambda$ consists of fixed points of the Hamiltonian flow,
and $\Lambda$ is $(a,\epsilon)$-normally hyperbolic with any $\epsilon >0$.
}
\end{exam}

\section{Applications to thermodynamics}\label{sec-therm}

\subsection{The starting question}
Let us introduce the following glossary.
A contact manifold $(\Sigma,\xi)$ is a {\it thermodynamic phase space}, a contact form $\lambda$ is {\it the Gibbs form}, trajectories of a contact Hamiltonian flow $\varphi^t: \Sigma \to \Sigma$ generated by a contact Hamiltonian $H: \Sigma\to \R$ are thermodynamic processes. Legendrian submanifolds $\Lambda \subset \Sigma$ play the role of equilibrium submanifolds.

Let $H : \Sigma \to \R$ be a contact Hamiltonian.
A {\it system of non-equilibrium thermodynamics} is a pair $(H,\Lambda)$ where
$\Lambda$ is a Legendrian submanifold lying in $M:= \{H=0\}$. One readily checks
that the contact flow $\{\varphi^t\}$ of $H$ preserves both $M$ and $\Lambda$.

\begin{question}[Stability problem of non-equilibrium thermodynamics] Given a system
$(H,\Lambda)$ and a subset $X \subset \Sigma$ of the phase space, does
there exist an initial condition $x \in X$ whose trajectory in the thermodynamic process generated by $H$ asymptotically converges to the equilibrium submanifold $\Lambda$?
\end{question}

Addressing this question, we tacitly assume local stability:  for $x$ lying in a small neighbourhood of $\Lambda$, the trajectory $\gamma:= \{\varphi^t x\}$  asymptotically converges to the equilibrium submanifold $\Lambda$ (cf. Section~\ref{sec-loc}).
Methods of contact topology provide a passage from local to global (see
in Theorem~\ref{thm-ising} below) which highlights our take on Prigogine's question
stated in the introduction.

\subsection{Warm-up: Newton's law of cooling} \label{subsec-cooling}
 Here we are motivated by \cite[p.45]{Haslach}. The coordinates $p,q,z$ in the thermodynamic phase space $\R^3$
correspond to the temperature, the entropy, and the internal energy, respectively.
Suppose that in the equilibrium
the internal energy of a thermodynamic system is given by $z=\phi(q)$.  (For instance, one can take $\phi(q) = e^q$, which corresponds to a constant volume ideal gas). The  equilibrium temperature for a given entropy $q$ equals $\phi'(q)$. The equilibrium states of our system form a Legendrian submanifold $\Lambda$ given by \eqref{eq-therm-lambda}.

Consider a non-equilibrium perturbation of the system whose equilibrium entropy and temperature are $\sigma$ and $\tau= \phi'(\sigma)$, respectively. We propose to describe the relaxation dynamics by a contact Hamiltonian
\[
H(p,q,z) = -a\big(z-\phi(q)\big) + b\big(p- \phi'(q)\big)(q-\sigma),\; a> b >0,\; \sigma >0.
\]
Introduce new coordinates
\color{black}
\[
P := p-{\phi'}(q),\ Z:= z-\phi(q),\ Q:= q-\sigma.
\] \color{black}
The change of coordinates preserves the contact form $\lambda$:
\[
\lambda = dz-pdq=dZ-PdQ.
\]
In the new coordinates $(P,Q,Z)$, the contact Hamiltonian is written as
\[
H(P,Q,Z)=-aZ +bPQ,
\]
and hence the contact Hamiltonian system reads as follows:
\[
\dot{Q}= -bQ,\ \dot{P} = -(a-b)P,\ \dot{Z} = -aZ.
\]
It follows that every trajectory of the system converges, as time goes to $+\infty$, to the point $(\phi'(\sigma), \sigma, \phi(\sigma)) \in \Lambda$. Let us note that the entropy increases along the trajectories with the initial condition $q(0) < \sigma$. Thus,  these trajectories are physically feasible.

The time evolution of the temperature $p(t)$ is given by
\[
\dot{p} = -(a-b)\big(p-\phi'(q)\big) - b\phi''(q)(q-\sigma).
\]
Since $q(t) = \sigma + o(1)$ as $t \to +\infty$,
this equation can be rewritten as
\[
\dot{p} \approx -(a-b)(p-\tau),
\]
{where $\tau := \phi' (\sigma)$ is the equilibrium temperature.}
 This is {\it Newton's law of cooling}.

Now set the parameter $b$ in the definition of $H$ to be $0$, cf. \cite{Haslach}.
Then for every point $(p,q,z)$ the integral trajectory originating at that point converges, as $t\to +\infty$, to $(\phi'(q),q, \phi(q))$, with $q(t)=q(0)$ being constant
for all $t$. Thus the Hamiltonian
\begin{equation}\label{eq-Ham-isentropic}
H(p,q,z) = -a\big(z-\phi(q)\big)
\end{equation} describes an isentropic
 relaxation (i.e. a relaxation with constant entropy) for all initial conditions. {(Isentropic processes may happen in an ideal thermodynamic system approximating a real one).}

In fact, we can modify this example in order to achieve that the entropy is increasing,
but the total increment of the entropy is arbitrarily
small for almost all initial conditions. To this end fix $a, \epsilon, N >0$ with $\epsilon N < a$ and put
\[
H(p,q,z) = -a\big(z-\phi(q)\big) - \epsilon \big(p- \phi'(q)\big)\sin^2(Nq).
\]
Introduce new coordinates:
\color{black}
\[
P := p-{\phi'}(q),\ Z:= z-\phi(q),\ Q:= q.
\]\color{black}
The change of coordinates preserves $\lambda$. In the new coordinates
\[
H (P,Q,Z) = -aZ - \epsilon P\sin^2(NQ),
\]
and the contact Hamiltonian system looks as follows:
\[
\dot{P} = - \big(a + \epsilon N \sin (2NQ)\big)P,\ \dot{Q}= \epsilon \sin^2(NQ),\ \dot{Z}= -a Z.
\]
It follows that if $q(0)$ lies in the interval $(\sigma_k,\sigma_{k+1})$ with  $\sigma_k:= \pi k/N$,  $k \in \Z$, the entropy
$q(t)$ strictly increases and converges to the right endpoint $\sigma_{k+1}$ as $t \to +\infty$.
Since for large $t$ we have $q(t)= \sigma_{k+1} + o(1)$ and since
$\sin (2N\sigma_{k+1})=0$, one readily gets that  the equation on the temperature $p(t)$ has a form $\dot{p} \approx -a(p-\tau_{k+1})$,
with $\tau_{k+1}:= \phi'(\sigma_{k+1})$ being the limiting temperature in the equilibrium. Thus we have again arrived at Newton's law of cooling.

Note that $|\sigma_{k+1}-q(t)| < \pi/N$.
Thus, taking larger and larger $N$ we can make the total increment of the entropy arbitrarily
small for almost all initial conditions.

\subsection{Thermodynamics of the Ising model} \label{subsec-ising-1}
Consider the system of spins located in the points of the integer lattice of any dimension
in a homogeneous external  magnetic field  \cite{Glauber, SuzukiKubo}.
Consider the standard contact space $\R^3 (p,q,z)$, now with $z$ being the free energy of the system, $p$ the magnetization
(i.e. the mean spin direction), and $q$ the magnitude of the external homogeneous magnetic field.
{The free energy of the system can be expressed in terms of the magnetization and the magnetic field by the following formula \cite[formula (3.1.6)]{Mussardo}:
\begin{equation} \label{eq-equiHam}
-z=F:= -\phi_\beta(q+bp) + \frac{1}{2}bp^2,
\end{equation}
where
$\phi_\beta(u) = \beta^{-1}\ln 2\cosh (\beta u)$. Here \color{black}$b \geq 0$ \color{black} and $\beta >0$ are real parameters:
$\beta$ is the inverse temperature, and $b$ is determined by the strength
of the interaction and the geometry of the model.

In the equilibrium
$p = - \partial F/\partial q = \phi_\beta'(q+bp)$ ({\it the self-consistency equation} -- see \cite[formula (3.1.7)]{Mussardo}).}
The Legendrian submanifold formed by the equilibrium states is given by
\[
\Lambda_{b,\beta}:= \{p= \phi_\beta'(q+bp),\ z= \phi_\beta(q+bp) - bp^2/2\}.
\]
Denote by $L_{b,\beta}$ its projection to the $(p,q)$-plane.

In what follows we discuss and compare two scenarios of relaxation of the Ising model.

\medskip
\noindent
{\bf First (mean field Glauber) scenario:} This is the standard model described by Glauber dynamics, which, roughly speaking, is a sophisticated Markov process on the space of all possible spin configurations \cite{Glauber}. An analysis of this dynamics, combined with the
mean field approximation, yields
the following time evolution of the magnetization (see  equation (4.3) in \cite{SuzukiKubo}):
\begin{equation}\label{eq-kubo}
\dot{p} = {-c p + c\phi_\beta'(q+bp)},\; c >0.
\end{equation}
Here $q$ and $z$ remain constant and $c$ is a time scale parameter.
We will outline a derivation of this equation based on Glauber dynamics in Section~\ref{subsec-more-on-Glauber-dynamics}.

We are interested in the asymptotic behaviour of the solution of \eqref{eq-kubo}
for a fixed $q$ with $p(0)=p_0$:
\begin{equation}
\label{eq-limp2-a}
p^{(I)}_\infty(q,p_0) := \lim_{t \to +\infty} p(t),
\end{equation}
where upper index $(I)$ indicates that we are discussing the first model of relaxation.

Consider the vector field $v^{(q)}_{b,\beta}(p)= c\left(-p + \tanh\left(\beta(q+bp)\right)\right)$ on $\R$.

\medskip
\noindent
{\sc Case A:} $b\beta < 1$. In this case a direct calculation shows that $L_{b,\beta}$ is the graph of a smooth
odd function $r_{b,\beta}(q)$. The point $r_{b,\beta}(q) \in \R$ is the unique
attracting zero of the vector field $v^{(q)}_{b,\beta}$. We record that
\begin{equation}\label{eq-limp-1-A}
p^{(I)}_\infty(p_0,q) = r_{b,\beta}(q)\;\; \forall p_0 \in \R\;.
\end{equation}

\medskip
\noindent
{\sc Case B:} $b\beta > 1$. In this case a direct calculation shows that there exists $a_{b,\beta} > 0$ and a pair
of functions $r^+_{b,\beta} : [-a_{b,\beta},+\infty) \to (0,1)$ and
$r^-_{b,\beta}(q):= - r^+_{b,\beta} (-q)$, $q \in (-\infty, a_{b,\beta})$ so that the graphs
of $r^\pm$ lie in  $L_{b,\beta}$. Moreover, for $|q| \geq a_{b,\beta}$, the vector field
$v^{(q)}_{b,\beta}$ contains unique attracting zero at $r^-_{b,\beta}(q)$ for negative $q$
and at $r^+_{b,\beta}(q)$ for positive $q$, respectively. For $|q| <  a_{b,\beta}$, in
addition to two attractive zeroes at $r^{\pm}_{b,\beta}(q)$, this field has the third, repelling zero lying between them. This zero is given by an odd decreasing function
$s_{b,\beta}$ defined on the interval $(-a_{b,\beta},a_{b,\beta})$ so that
$$L_{b,\beta} = \text{graph} (s_{b,\beta}) \cup \text{graph}(r^+_{b,\beta})  \cup \text{graph}(r^-_{b,\beta}),$$ see Figure 3.  It will be convenient to put $s_{b,\beta}(q) = -\infty$ for $q > a_{b,\beta}$ and $s_{b,\beta}(q) = +\infty$ for $q < -a_{b,\beta}$.

\begin{figure}[h]
\label{fig-vsp}
\caption{$\beta=1, b=6$}
\includegraphics[width=0.5\textwidth]{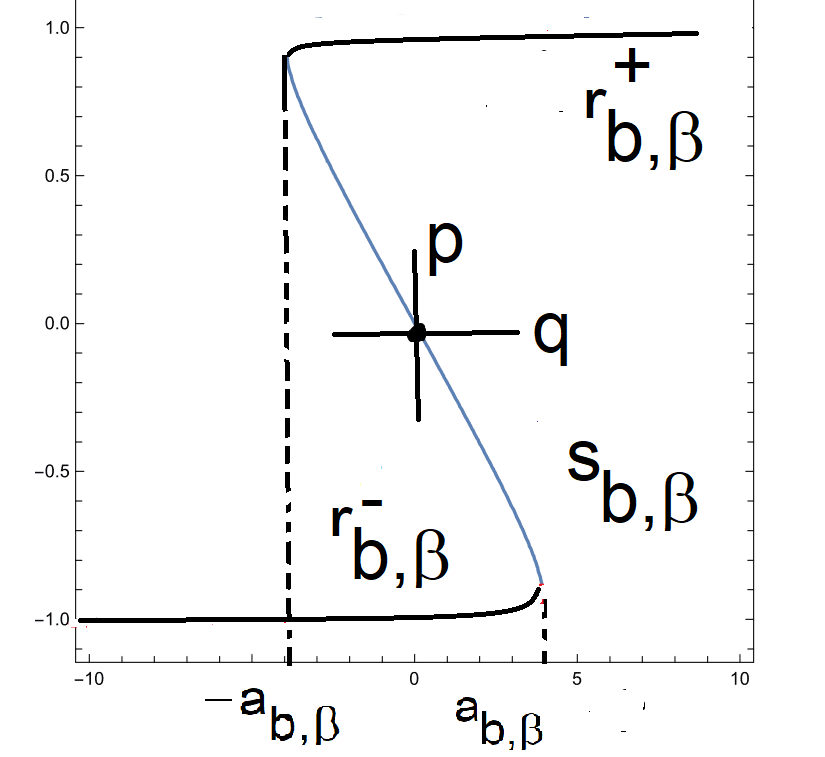}
\centering
\end{figure}

The upshot of this discussion is that
\begin{equation}\label{eq-limp-1-2a}
p^{(I)}_\infty(p_0,q) = r^+_{b,\beta}(q),\;\;\text{for}\;\; q \geq -a_{b,\beta}(q)\;\;\text{and}\;\;  p_0 \in \left(s_{b,\beta}(q),+\infty\right)
\end{equation}
and
\begin{equation}\label{eq-limp-1-2b}
p^{(I)}_\infty(p_0,q) = r^-_{b,\beta}(q),\;\;\text{for}\;\; q \leq a_{b,\beta}(q)\;\;\text{and}\;\;  p_0 \in \left(-\infty,s_{b,\beta}(q)\right).
\end{equation}
Finally, we have an unstable equilibrium
\begin{equation}\label{eq-limp-1-2c}
p^{(I)}_\infty\left(s_{b,\beta}(q),q\right)= s_{b,\beta}(q).
\end{equation}
Let us emphasize that in Case B the map $p^{(I)}$ is discontinuous,
which  manifests the phase transition of the Ising model \cite{Mussardo}.

\medskip
\noindent
{\bf Second (contact) scenario:}
In the second model, motivated by \cite[Example 2.5]{Goto-JMP2015},  relaxation is governed by the contact Hamiltonian system with the Hamiltonian
\begin{equation}\label{eq-relax}
h_{b,\beta} (p,q,z) = c(-z- F(p,q)) = c\big(-z + \phi_\beta (q+bp) - bp^2/2\big).
\end{equation}
Thus the magnetization $p$ satisfies equation \eqref{eq-kubo} as in the previous model.
The key difference, however is that in the first model the magnetic field $q$ is conserved,
while in the second one, as we shall see immediately, {\it the effective magnetic field} $q+bp$ is conserved
(also see the discussion at the end of Section~\ref{subsec-more-on-Glauber-dynamics}).

In order to understand the contact dynamics of $h_{b,\beta}$ make the following change
of variables:
\[
P:=p,\ Q:= q+bp,\ Z:= z+bp^2/2.
\]
The change of variables preserves the contact form $\lambda$:
\[
\lambda = dz-pdq = dZ-PdQ.
\]
In the new coordinates,
\[
\Lambda_{eq}= \big\{P= \phi_\beta'(Q),\ Z=\phi_\beta(Q)\big\},
\]
and the contact Hamiltonian $h_{b,\beta}$ is given by
\begin{equation} \label{eq-contHam-b}
c\big(-Z + \phi_\beta(Q)\big)\;.
\end{equation}
The corresponding contact Hamiltonian system looks as follows:
\[
\dot{P} = -cP + c\phi_\beta'(Q),\ \dot{Q} = 0,\ \dot{Z} = c\big(-Z+\phi_\beta(Q)\big).
\]
In particular, the equation for $\dot{P}$ matches \eqref{eq-kubo}.
The system can be solved explicitly: the trajectory $\big( P(t), Q(t), Z(t)\big)$ passing at $t=0$ through $(P_0 = p_0,Q_0 = q_0 + bp_0,Z_0 = z_0 + bp_0^2/2)$ is given by
\[
\Big(e^{-ct}\big(P_0-\phi_\beta'(Q_0)\big) + \phi_\beta'(Q_0),Q_0, {e^{-ct}\big(Z_0 - \phi_\beta (Q_0)\big) +\phi_\beta (Q_0)} \Big).
\]
In particular, the trajectory converges, as $t\to +\infty$, to the equilibrium state $\big(\phi_\beta' (Q_0), Q_0, \phi_\beta (Q_0)\big) =  \big(\phi_\beta' (q_0 + bp_0), q_0 + bp_0, \phi_\beta (q_0 + bp_0)\big)\in \Lambda_{b,\beta}$, exhibiting the relaxation process.
Expressing the integral trajectory in the original coordinates $(p,q,z)$ as $\big( p(t),q(t),z(t)\big)$ we see that
\[
Q(t) = q(t) + bp(t) = q_0 + bp_0.
\]
In particular,
\begin{equation}\label{eq-limp}
p^{(II)}_\infty(p_0,q_0):= \lim_{t \to +\infty} p(t)= \tanh \left(\beta(q_0 + bp_0)\right)\;.
\end{equation}
Furthermore, the magnitude $q$ of the external magnetic field changes during the relaxation process:
\[
q(t) = q_0 + b\big(p_0 -p(t)\big),
\]
and converges as $t \to +\infty$ to
\begin{equation}\label{eq-qlim}
q^{(II)}_\infty(p_0,q_0) = q_0 + b\big(p_0 -\phi_\beta'(q_0+bp_0)\big)\;.
\end{equation}

Let us mention that contact Hamiltonian \eqref{eq-contHam-b}
(and its counterpart \eqref{eq-Ham-isentropic} above) appear in the context of contact thermodynamics in \cite[formula (89)]{Bravetti} and \cite[Theorem 4.3]{Goto-JMP2015}.

\medskip
\noindent
{\sc Comparison between the scenarios:} \color{black} In case $b =0$, i.e., in the absence of the interaction between the spins (see \cite[Example 2.5]{Goto-JMP2015}), both scenarios
coincide. This is however not the case when $b > 0$. \color{black}
In the first model the external magnetic field $q$ is constant, so the Glauber equation \eqref{eq-kubo} can be considered as a family of ODEs
describing evolution of the magnetization $p$,  depending on $q$ as a parameter. In the second, contact, model $q$ becomes an independent variable which enables us to view dynamics in the thermodynamic phase space. In this approach, however, the effective magnetic field $Q= q+bp$
is an integral of motion so mathematically one can again view the evolution as a family of ODEs
on the magnetization depending on $Q$ as a parameter. In what follows we shall explore the dynamics
of {\it perturbations} of the contact Hamiltonian  $h(P,Q,Z)$ so that $Q$ is no more conserved and
the dynamics becomes genuinely three-dimensional.

Furthermore, by using formula \eqref{eq-qlim}, one can show that
the two scenarios agree up to a small error, as $t\to +\infty$, whenever the initial conditions $(p_0,q_0)$
satisfy, with the following assumptions for a sufficiently small $\delta>0$ and a sufficiently large $K >0$:
\begin{itemize}
\item $|p_0 -\phi_\beta'(q_0+bp_0)| \ll \delta$, i.e., the process starts at most $\delta$-far from the equilibrium;
    \item $|q_0| \geq K$, where in Case B we assume that $K > a_{b,\beta}$, i.e.
the process starts far from the ``phase transition region".
\end{itemize}
Then one can show that $|p^{(I)}(p_0,q_0) - p^{(II)}(p_0,q_0)| \to 0$ as
$\delta \to 0$ and $K \to +\infty$.  An important feature of the contact dynamical model is that it mollifies the discontinuity of the relaxation process in the phase transition region:
the map $p^{(II)}$ is continuous.

Let us mention that after the first version of this paper appeared, Goto \cite{Goto-new}
suggested an interesting contact geometric approach to the first (i.e., the mean field Glauber) scenario incorporating the phase transitions. This line of research was further
advanced in \cite{GLP}.

\subsection{Stability}
\label{subsec-stability-ising}
In what follows we fix \color{black}$b\geq 0$\color{black}. Consider an Ising model with the temperature $1/\beta$ whose equilibrium is given by the Legendrian submanifold
\[
\Lambda_{b,\beta}:= \{p= \phi_\beta'(q+bp),\ z= \phi_\beta(q+bp) - bp^2/2\}.
\]
{Note that the contact Hamiltonian
$h_{b,\beta}$, defined by \eqref{eq-relax},
vanishes on
$\Lambda_{b,\beta}$, and the corresponding flow preserves this submanifold point-wise.

Suppose that parameters $b,\beta$ suddenly changed to $a,\alpha$, so that the system moved
to a new equilibrium submanifold $\Lambda_{a,\alpha}$. We shall establish, under certain assumptions, existence of a relaxation process starting at $\Lambda_{a,\alpha}$ and converging to $\Lambda_{b,\beta}$ for a class of Hamiltonians which coincide with $h_{b,\beta}$ (see \eqref{eq-relax} above) {\it near} the equilibrium $\Lambda_{b,\beta}$, but which can
substantially deviate from  $h_{b,\beta}$ away from the equilibrium.
To facilitate further discussion, introduce the {\it effective magnetic field}
$Q= q+bp$ and {\it deviations from the equilibrium}  $\Lambda_{b,\beta}$ by setting
\color{black}
$$P:= p- \phi'_{\beta}(Q),\; Z:= z- \phi_{\beta}(Q)+ bp^2/2\;.$$\color{black}
Note that the contact form is preserved in these coordinates,
$dZ-PdQ=dz-pdq$. Furthermore, $h_{b,\beta}(P,Q,Z) = -cZ$, and $\Lambda_{b,\beta}=\{Z=P=0\}$.

In the new coordinates, the Legendrian $\Lambda_{a,\alpha}$ is given by the equations
\begin{equation}\label{eq-lambda-1vsp}
P - \tanh\left(\alpha Q+\alpha(a-b)\left(P+ \tanh(\beta Q)\right)\right)+ \tanh(\beta Q)=0\;,
\end{equation}
\begin{equation}\label{eq-lambda-2vsp}
\begin{split}
Z- \alpha^{-1}\log \left( 2 \cosh \left(\alpha Q+\alpha(a-b)(P+ \tanh(\beta Q)) \right) \right) \\ +
\beta^{-1}  \log\left((2\cosh(\beta Q)\right) + 0.5(a-b)(P+ \tanh(\beta Q))^2=0\;.
\end{split}
\end{equation}

\medskip
Next we introduce a class of contact Hamiltonians we are going to deal with.
Let $H$ be a complete contact Hamiltonian on $\Sigma=\R^3(P,Q,Z)$
periodic in $Q$ with period $\tau>0$.  Denote by $T: \Sigma \to \Sigma$ the shift
$Q \mapsto Q+\tau$. Let a submanifold $M \subset \Sigma$ be the union of a finite number of the connected components of the nodal set $\{H=0\}$ so that $M$ separates $\Sigma$ into two open parts, $\Sigma_-$ and $\Sigma_+$. Assume that $H$ is strictly positive on $\Sigma_+$, but in general it is allowed to change sign on $\Sigma_-$. We also assume that $M$ and $\Sigma_{\pm}$ are invariant under $T$. Suppose that the zero section $\Lambda_{b,\beta}$ lies in $M$. Assume furthermore that
\begin{itemize}
\item [{(i)}] There exists $\kappa_1 > 0$ such that $\partial H/\partial Z \leq 0$
on $\{0 < H <\kappa_1 \} \cap \Sigma_+$.
\item [{(ii)}] There exists $\kappa_2 > \kappa_1 > 0$ such that $\partial H/\partial Z  \leq 0$ on $\{H \geq \kappa_2\} \cap \Sigma_+$.
\item [{(iii)}] $\partial H/\partial Z <0$ near $\Lambda_{b,\beta}$.
\item[{(iv)}] $\Lambda_{b,\beta}$ is a local attractor of the contact flow of $H$ in
its $T$-invariant neighbourhood.
\end{itemize}

\begin{defin}
\label{def-admissible-contact-Ham}
{\rm We call such Hamiltonians $H$ {\it admissible}.
}
\end{defin}

\medskip\noindent For instance $H(P,Q,Z) = -cZ$ is admissible.

\begin{thm} \label{thm-ising} Let $H(P,Q,Z)$ be any admissible
Hamiltonian with the period $\tau$. Suppose that either (I) or (II) holds.
\begin{itemize}
\item[{(I)}] \color{black}
\begin{equation}
\label{eq-ab1}
a > b \geq 0, \alpha > \beta > 0  \;,
\end{equation}\color{black}
and there exists $Q_0$ such that $H$ is negative in a neighbourhood
of $(0,Q_0,(a-b)/2)$.
\item[{(II)}]
\begin{equation}
\label{eq-ab2}
a < b\;.
\end{equation}
\end{itemize}
Then the flow of $H$ necessarily possesses a trajectory, lying
in $\Sigma_+$
which starts at $\Lambda_{a,\alpha} \cap \Sigma_+$ and asymptotically converges to $\Lambda_{b,\beta}$.
\end{thm}

\medskip
\begin{rem}\label{rem-per}
{\rm The periodicity assumption will enable us to work with (compact!)
Legendrian circles instead of non-compact Legendrian lines.
We need compactness in order to apply the results on the existence of asymptotic
trajectories from Section \ref{sec-global} above to the proof of Theorem \ref{thm-ising}.
It is unclear to us whether the theorem remains true without this assumption.}
\end{rem}

\begin{proof} {\bf I.} Assume first that \color{black}$a > b \geq 0, \alpha > \beta >0$.\color{black}

\medskip
\noindent
{\sc Step 1:} Let us analyze the chords of the Reeb flow starting on $\Lambda_{a,\alpha}$ and
ending on $\Lambda_{b,\beta}$. Note that $P=0$ along such chords. Substituting $P=0$
into equation \eqref{eq-lambda-1vsp} we readily see that
$$\alpha(a-b)\tanh (\beta Q) = (\beta-\alpha) Q\;.$$
Since $a > b, \alpha > \beta$, the right and the left hand sides of this equation
have different signs unless $Q = 0$. Thus $Q=0$ is the unique solution. The value
of $Z$ corresponding to $P=Q=0$ equals $\log 2 (\alpha^{-1} -\beta^{-1}) < 0$.
It follows that there exists a unique chord of the Reeb flow from $\Lambda_{a,\alpha}$
to the zero section $\Lambda_{b,\beta}$. Furthermore, differentiating {the implicit expression $P(Q)$ given by \eqref{eq-lambda-1vsp}}, one readily gets that
\begin{equation}
\label{eqn-dP-dQ}
\left(1-\alpha(a-b)\right)\frac{dP}{dQ}(0) = \alpha - \beta +\alpha\beta (a-b).
\end{equation}
By the hypothesis of the theorem, the right-hand side in this equality is always positive and therefore
$dP/dQ(0)\neq 0$, which immediately
yields non-degeneracy condition \eqref{eq-nondeg}.

\medskip\noindent{\sc Step 2:} Let us explore the behavior of $\Lambda_{a,\alpha}$
as $Q \to \pm \infty$. One can readily check using \eqref{eq-lambda-1vsp}, \eqref{eq-lambda-2vsp} that then $P \to 0$ and $Z \to (a-b)/2$. Consider the
shift $T: Q \mapsto Q+\tau$. Denote by $\cV$ a rectangular neighbourhood of $x_0:= (0,Q_0, (a-b)/2)$
where $H$ is negative. Then $\Lambda_{a,\alpha}$ intersects $T^{\pm k/2}\cV$ for a sufficiently large even integer $k$, ``entering" these neighbourhoods via the left vertical edge of $T^{- k/2}\cV$ and the right vertical edge of $T^{k/2}\cV$, respectively.

\medskip\noindent{\sc Step 3:} Put $\Pi= \{|Q-Q_0| \leq k\tau/2\}$,
and consider the curve $\Lambda_{a,\alpha} \cap \Pi$ with the endpoints
lying in $T^{\pm k/2}\cV$. Modify this curve near the endpoints and
get a new Legendrian submanifold $K_{a,\alpha} \subset \Pi $
such that $K_{a,\alpha}$ coincides with $\{p=0, Z= (a-b)/2\}$ near $T^{\pm k/2}x_0$, and $K_{a,\alpha}$ coincides with $\Lambda_{a,\alpha} \cap \Pi$ on $\Sigma_+ \cap \Pi$.
Let us mention that this modification can be made smooth in $a,\alpha$ for fixed $b,\beta$,
with $K_{b,\beta}= \Lambda_{b,\beta}$ being the zero section, $x_0:= (0,Q_0, (a-b)/2)$,
and $k$ fixed and large enough.

\medskip\noindent{\sc Step 4:} Let us make the following observation.
Fix $k \in \N$, put $S^1= \R/k\tau Z$, and define $\Sigma':= T^*S^1 (P, Q\mod k\tau) \times \R(Z)$. Let $\pi:\Sigma \to \Sigma'$ be the natural projection.
Any $\tau$-periodic admissible Hamiltonian $H$ descends to a Hamiltonian $H'$ on $\Sigma'$. Put $\Lambda'_{b,\beta}= \pi(\Lambda_{b,\beta})$, $M'=\pi(M)$, and $\Sigma'_\pm = \pi(\Sigma_\pm)$. Then $(H',\Lambda'_{b,\beta})$ satisfy assumption $\clubsuit$
from Section~\ref{subsec-as-1}. Note also that the strip $\Pi$ defined in the previous
step is a fundamental domain of the covering $\pi$.

With this notation, the Legendrian $K_{a,\alpha} \subset \Pi$ constructed in Step 3
defines a closed Legendrian $K'_{a,\alpha}  \subset \Sigma'$.  By Step 1, there is a unique {non-degenerate} Reeb chord from $K'_{a,\alpha}$ to $\Lambda'_{b,\beta}$.  Furthermore,
the family of Legendrians $K'_{c, \delta}$ with $c \in [b,a]$, $\delta \in [\beta,\alpha]$
defines a Legendrian isotopy between $K'_{a,\alpha}$ and the zero section.  Thus we are in
the situation of Example~\ref{exam-onechord}, and it follows that the
pair $(K'_{a,\alpha}, \Lambda'_{b,\beta})$ is interlinked. Therefore, by Theorem~\ref{thm-A}, there exists a trajectory $\gamma' \subset \Sigma'_+$ of $H'$ which starts on $K'_{a,\alpha}$ and converges to  $\Lambda'_{b,\beta}$.

Lifting $\gamma'$ to the universal cover $\R^3 (P,Q,Z) \to (P, Q\mod \tau,Z)$,
we get a trajectory $\gamma \subset \Sigma_+$ of $H$ starting on the lift
$\bigcup_{j \in \Z} T^{kj}(K_{a,\alpha})$ of $K'_{a,\alpha}$ and converging to $\Lambda_{b,\beta}$. Here we are using that $\Lambda_{b,\beta}$ is invariant under the group of the deck transformations of the cover. Since $\gamma \subset \Sigma_+$, and since $K_{a,\alpha}$ coincides with $ \Lambda_{a,\alpha}$ in $\Sigma_+$ (here we use that $H$ is negative near $x_0$), we get that this trajectory in fact starts
on $\Lambda_{a,\alpha}$. This completes the proof of the theorem under assumption \eqref{eq-ab1}.

\medskip
\noindent
{\bf II.} Assume now that $a <b$.  Let $\cU$ be the attraction basin of the zero section, i.e.,
the collection of all points of $\Sigma$ whose semi-trajectories under the contact flow of $H$ converge to the zero section $\Lambda_{b,\beta}$. Since $\Lambda_{b,\beta}$ is a local attractor,
$\cU$ is open. Consider a Legendrian submanifold $$K:= \{P=0, Z= (a-b)/2\} \subset \Sigma_+\;.$$
Let $\pi: \Sigma \to \Sigma' = T^*S^1 \times \R$ be the natural projection to the quotient
of $\Sigma$ by the shift $T: Q \mapsto Q+\tau$. Since  the pair consisting of $\pi(K)$ and the zero section is interlinked (see Example~\ref{exam-1-jet-space-0-section-and-its-Reeb-shift-homol-bonded} above),
Theorem~\ref{thm-A} applied to these Legendrians yields $K \cap \cU \neq \emptyset$.
Let $\cV \subset \cU$ be a neighbourhood of a point $v \in K$. Then $T^k(\cV) \subset \cU$
for all $k \in \Z$. But $\Lambda_{a,\alpha}$ is asymptotic to $K$ as $Q \to \infty$. Therefore,  $\Lambda_{a,\alpha} \cap T^k(\cV) \neq \emptyset$ for all $k$ sufficiently large. This yields existence of the desired chord, and hence completes the proof of the theorem under assumption \eqref{eq-ab2}.
 \end{proof}

\begin{rem}{\rm If $a>b$, $\alpha > \beta$, the front projection of the curve $\Lambda_{a,\alpha}$ to the
$(Z,Q)$-plane is subject to a ``phase transition" at
\[
1-\alpha (a-b) =0.
\]
Indeed, the sign of $dP/dQ(0)$ -- which, by \eqref{eqn-dP-dQ}, coincides with the sign of $1-\alpha (a-b)$ --  is responsible for the convexity/concavity
of the front projection of $\Lambda_{a,\alpha}$ at the initial point of the Reeb chord from $\Lambda_{a,\alpha}$ to $\Lambda_{b,\beta}$, see Figures 4 and 5 (the projection of the Reeb chord connecting $\Lambda_{a,\alpha}$ with the zero section
(i.e. with $\Lambda_{b,\beta}$), is shown in the pictures by the dashed line).}
This ``phase transition" of the front projections in the case $b=0$ corresponds to the physical phase transition in the Ising model well-known in statistical physics (see e.g. \cite{Mussardo}, the discussion after (3.1.8)).

\begin{figure}[!tbp]
  \centering
  \begin{minipage}[b]{0.4\textwidth}
    \includegraphics[width=\textwidth]{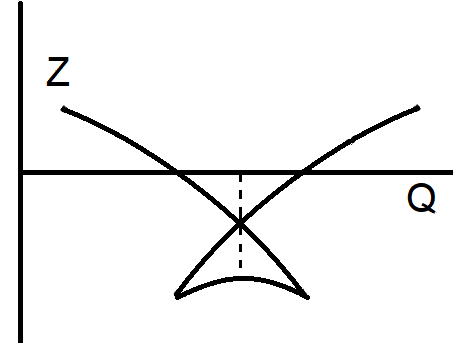}
    \caption{$\beta=0.4, \alpha= 1, b=1.5, a = 4$}
  \end{minipage}
  \hfill
  \begin{minipage}[b]{0.4\textwidth}
    \includegraphics[width=\textwidth]{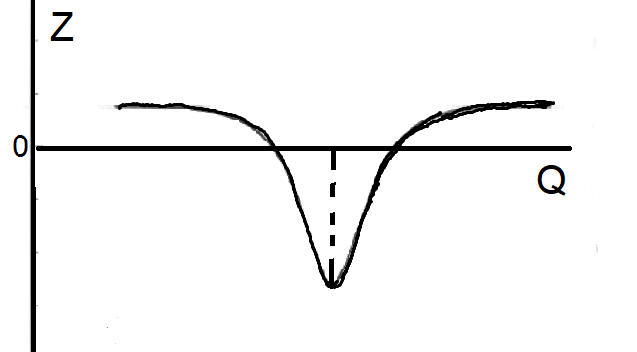}
    \caption{$\beta=0.1,\alpha= 0.2,b=2,a = 4$}
  \end{minipage}
\end{figure}

\end{rem}
\begin{rem}
{\rm In the remaining case $a>b$, $\alpha < \beta$, the conclusion of Theorem~\ref{thm-ising}
is in general wrong. Indeed, an elementary (albeit cumbersome) argument shows that in this
case $\Lambda_{a,\alpha}$ lies in the half-space $\{Z > \delta\}$ with some $\delta >0$.
Assume that $H=-cZ$ on $\{Z < \delta/2\}$, so that $\Sigma_-= \{Z >0\}$. Recall that we have
no restrictions on $H$ in $\Sigma_-$. Assume that $H$ vanishes on $\{Z =\delta\}$,
having it as a regular component of the nodal set. Then the latter hypersurface is invariant under the flow and separates $\Lambda_{a,\alpha}$ and $\Lambda_{b,\beta}$,
ruling out the existence of asymptotic chords.
}
\end{rem}

\subsection{Glauber dynamics}
\label{subsec-more-on-Glauber-dynamics}
In this section we recall a derivation of the Glauber dynamics for the
mean field Ising model. In the course of our exposition, which mostly follows \cite{Glauber,SuzukiKubo}, we elaborate on some details of the mean field approximation. Then, in  next section we propose a modification of the Glauber dynamics which gives rise to admissible Hamiltonians appearing in Theorem~\ref{thm-ising} above.

We are going to work with a set $G$ of $N$ spins. Write $\cS$ for the space of configurations $\sigma: G \to \{-1,1\}$. We consider $\cS$ as a finite graph, where $\sigma$ and $\sigma'$ are connected by an edge whenever there exists unique element $g \in G$ such that $\sigma(g) = -\sigma'(g)$ and $\sigma(h) = \sigma'(h)$ for all $h\in G$, $h \neq g$. In this case we write
$\sigma' =: \bar{\sigma}_g$. Define a Markov chain on $\cS$, where the  probability
of transition from $\sigma$ to $\bar{\sigma}_g$ is denoted by $\hat{w}_g(\sigma)$.
We say that a configuration $\sigma$ {\it flips} at a point $g$ of the lattice if the spin
$\sigma(g)$ changes the sign at a given step of the Markov evolution. Note that with our notation $\overline{(\bar{\sigma}_g)}_g = \sigma$, and therefore the probability
of the transition from $\bar{\sigma}_g$ back to $\sigma$ equals $\hat{w}_g(\bar{\sigma}_g)$.

The Hamiltonian \footnote{Here the word ``Hamiltonian" is used in the sense of statistical mechanics, as opposed to contact Hamiltonians we encounter throughout the paper. It models the energy of a configuration $\sigma$.} of the Ising model is given by
$$\cH(\sigma)= -q \sum_g \sigma(g) - \sum_{g\neq h}J_{gh} \sigma(g)\sigma(h)\;,$$
where $q$ is the external magnetic field and coefficients $J_{gh}$ describe the interaction
between the spins. In the equilibrium, the probability of a configuration $\sigma$ is
given by the Gibbs distribution
$$\varpi(\sigma) = \frac{e^{-\beta \cH(\sigma)}}{\mathcal Z}\;,$$
where $\mathcal{Z} = \sum_\sigma e^{-\beta \cH(\sigma)}$ is the partition function.
Following \cite{Glauber,SuzukiKubo},
we assume the following properties of the transition probabilities:

\smallskip
\noindent
(I) ({\it detailed balance}) $$\frac{\hat{w}_g(\sigma)}{\hat{w}_g(\bar{\sigma}_g)}=
\frac{\varpi(\bar{\sigma}_g)}{\varpi(\sigma)} = \frac{e^{-\beta \cH(\bar{\sigma}_g)}}{e^{-\beta \cH(\sigma)}}\;.$$

In other words, the Glauber Markov chain is a reversible Markov chain possessing the Gibbs distribution as the stationary one. The detailed balance condition yields that there
exists $\hat{c}(\sigma,g) >0$ such that
\begin{equation}\label{eq-detbal}
\hat{w}_g(\sigma) = \hat{c}(\sigma,g) \frac{\varpi(\bar{\sigma}_g)}{\varpi(\sigma)+ \varpi(\bar{\sigma}_g)}\;, \hat{w}_g(\bar{\sigma}_g)= \hat{c}(\sigma,g) \frac{\varpi(\sigma)}{\varpi(\sigma)+ \varpi(\bar{\sigma}_g)}\;.
\end{equation}
This means that the transition probability $\hat{w}_g(\sigma)$ is proportional
to the conditional probability with respect to the stationary distribution $\varpi$ of the following event:
{\it a configuration equals $\bar{\sigma}_g$ provided it coincides with $\bar{\sigma}_g$ outside $g$.} The probability $\hat{w}_g(\bar{\sigma}_g)$ of the transition from $\bar{\sigma}_g$ to $\sigma$ admits a similar interpretation.

Our next assumption is as follows:

\smallskip
\noindent
(II)({\it time scale parameter}) The proportionality coefficient $\hat{c}(\sigma,g)$ defined above
does not depend on $\sigma$ and $g$:
$$\hat{c}(\sigma,g) = \hat{c} >0\;,$$
for all $\sigma \in \cS, g \in G$.
 Let us mention that this yields
$\hat{w}_g(\sigma) + \hat{w}_g(\bar{\sigma}_g)= \hat{c}$ for all $\sigma \in \cS$ and $g \in G$. The parameter $\hat{c}$ appears in \cite[p.296]{Glauber} as the time scale on which all transitions take place. We shall see in  Remark~\ref{rem-rate} below
that $\hat{c}$ is related to the relaxation rate of the magnetization of the system towards the equilibrium.

\smallskip
It follows from (I) and (II) that
\begin{equation}
\label{eq-w}
\hat{w}_g(\sigma) = \frac{\hat{c}}{2}\left(1+\tanh (\beta \Delta/2)\right),\;
\end{equation}
where
\begin{equation}\label{eq-Delta}
\Delta= \cH(\sigma)-\cH(\bar{\sigma}_g) = -2\sigma(g)\left(q+ \sum_{h\neq g} J_{gh}\sigma(h)\right)\;.
\end{equation}

\begin{rem} {\rm It is instructive to compare \eqref{eq-w},\eqref{eq-Delta} with formula (2.7) in \cite{SuzukiKubo},
stating that (in our notation)
$$\hat{w}_g(\sigma) = \frac{\hat{c}}{2}\left(1 - \sigma(g)\tanh (\beta \Delta')\right)\;,$$
with (see (2.5) in \cite{SuzukiKubo})
$$\Delta' = q+ \sum_{h\neq g} J_{gh} \sigma(h)\;.$$
One readily sees that both expressions for $\hat{w}_g(\sigma)$ coincide since $\sigma(g) \in \{-1,+1\}$
and the hyperbolic tangent function $\tanh$ is odd.
}
\end{rem}

Consider the mean spin of a configuration $\sigma$,
$$m(\sigma):= \frac{1}{N}\sum_g \sigma(g)\;.$$
Given a probability measure $\pi$ on $\cS$, {\it the magnetization
of the system} is defined as the expectation
of $m$ with respect to $\pi$: $p = \mathbb{E}(m(\sigma))$.

Let us discuss now the time evolution of the Glauber Markov chain.
We start with the case of the discrete time. Introduce a parameter $\tau$
having a dimension of time, and assume that the transitions happen at time moments which are
positive integral multiples of $\tau$. Let $\pi(\cdot,t)$ be a $t$-dependent family of  probability measures on $\cS$ evolving under the Markov chain. We have the following master equation:
\[
\pi(\sigma, t + \tau) =\pi(\sigma,t)\Big( 1- \sum_{g \in G} \hat{w}_g(\sigma)\Big) + \sum_{g \in G}\pi(\bar{\sigma}_g,t) \hat{w}_g(\bar{\sigma}_g).
\]

Let us pass now to the continuous time $t$.
  \footnote{This passage is somewhat implicit in \cite{Glauber,SuzukiKubo}.
We thank the referee for explaining it to us.}  We consider $\tau$ as a small parameter,
and assume that the constant $\hat{c}$ from (II) is proportional to $\tau$, i.e.,
\begin{equation}\label{eq-c-vsp}
\hat{c} = \tau c\;,
\end{equation}
for some $c >0$. Then, according to \eqref{eq-w}, the transition probabilities $\hat{w}_g$ are given by $\hat{w}_g(\sigma) = \tau {w}_g(\sigma)$
with
\begin{equation}\label{eq-w-1}
{w}_g(\sigma) = \frac{{c}}{2}\left(1+\tanh (\beta \Delta/2)\right)\;,
\end{equation}
where $\Delta$ is defined in \eqref{eq-Delta}.
Rewrite the master equation in the form
\[
\frac{\pi(\sigma, t + \tau) - \pi(\sigma,t)}{\tau} = -  \sum_{g \in G} \pi(\sigma,t) {w}_g(\sigma) + \sum_{g \in G}\pi(\bar{\sigma}_g,t) {w}_g(\bar{\sigma}_g).
\]
Passing to the limit when $\tau \to 0$, we get
an ordinary differential equation
\begin{equation}\label{eq-master}
\frac{d}{dt} \pi (\sigma,t) = \sum_{g \in G} \Big(- \pi(\sigma,t){w}_g(\sigma) + \pi(\bar{\sigma}_g,t){w}_g(\bar{\sigma}_g)\Big).
\end{equation}
Denote by $\mathbb{E}_t$ the expectation with respect to the measure $\pi(\cdot,t)$.
In particular, given $g \in G$,
\[
\mathbb{E}_t \big(\sigma(g)\big) = \sum_\sigma \sigma(g)\pi(\sigma,t).
\]
We calculate
\[
\frac{d}{dt} \mathbb{E}_t \big(\sigma(g)\big) =  \frac{d}{dt} \sum_\sigma \sigma(g)\pi(\sigma,t) =  \sum_\sigma \sigma(g) \frac{d}{dt} \pi(\sigma,t).
\]
Now we use formula \eqref{eq-master} in order to compute the last term, re-denoting
the summation index in this formula by $h \in G$. Splitting the formula into the cases
$h=g$ and $h \neq g$, we write $\dot{p} = I_1 + I_2$, where $I_1$ and $I_2$ are defined as
follows:
$$I_1 = -\sum_{\sigma} \pi(\sigma,t)\sigma(g){w}_g(\sigma) + \sum_{\sigma} \pi(\bar{\sigma}_g,t)\sigma(g) {w}_g(\bar{\sigma}_g)= -2\mathbb{E}_t\big(\sigma(g){w}_g(\sigma)\big)\;,$$
where the last equality follows from  $\sigma(g) = -\bar{\sigma}_g(g)$,
and
$$I_2 = \sum_{h\neq g} \Big( -\sum_{\sigma} \pi(\sigma,t)\sigma(g){w}_h(\sigma) + \sum_{\sigma} \pi(\bar{\sigma}_h,t)\sigma(g) {w}_h(\bar{\sigma}_h)\Big)=$$ $$
\sum_{h\neq g} \Big( -\sum_{\sigma} \pi(\sigma,t)\sigma(g){w}_h(\sigma) + \sum_{\bar{\sigma}_h} \pi(\bar{\sigma}_h,t)\bar{\sigma}_h(g) {w}_h(\bar{\sigma}_h)\Big)=0\;.$$
Note that in the second equality in the last formula we have used that
$\bar{\sigma}_h(g)= \sigma(g)$ if $h \neq g$. We conclude that
\begin{equation}\label{eq-evol-vsp}
\frac{d}{dt} \mathbb{E}_t \big(\sigma(g)\big) = -2\mathbb{E}_t\big(\sigma(g){w}_g(\sigma)\big)\;.
\end{equation}
By definition, the magnetization of the system
equals
$$p(t) = \mathbb{E}_t\big(m(\sigma)\big) = \frac{1}{N}\sum_g  \mathbb{E}_t \big(\sigma(g)\big)\;.$$
Thus
$$\dot{p}(t) = -\frac{2}{N}\sum_g \mathbb{E}_t\big(\sigma(g){w}_g(\sigma)\big)\;.$$
Substituting formula \eqref{eq-w-1} for $w_g$ we get that
\begin{equation}\label{eq-prob}
\dot{p} = - cp + \frac{c}{N}\sum_g \mathbb{E}_t\left(\tanh \beta(q+ \sum_{h\neq g}J_{gh}\sigma(h))\right).
\end{equation}

\medskip\noindent{\sc Mean field approximation:} At this point we impose a non-rigorous physical assumption, called the {\it mean field approximation} (in \cite{SuzukiKubo} it is called ``molecular-field treatment"): with a high probability with respect to measure $\pi(\cdot,t)$, $$\sum_{g\neq h}J_{gh}\sigma(h) \approx bp,$$
where $b$ is a coefficient depending on the interaction $J_{gh}$, and $p$ is the magnetization. The essence of this assumption is that the sum $b^{-1} \sum_{h\neq g}J_{gh}\sigma(h)$ is taken over a ``large" number of neighbours, so an averaging takes place and its value is concentrated near the magnetization.  This holds, for instance, in the Curie-Weiss model when all spins interact with the constant strength: $J_{gh} \approx b/N$. It is plausible that Glauber's ODE can be derived rigorously in an appropriate thermodynamic limit of the Curie-Weiss model (thanks to S.~Shlosman for this comment), but apparently this is not yet done.

\medskip\noindent
With the mean field assumption we get
\begin{equation}\label{eq-evolution}
\dot{p} = -c\big(p - \tanh \left(\beta(q+bp)\right)\big).
\end{equation}
Thus, we derived equation \eqref{eq-kubo} above describing the Glauber evolution
in the mean field approximation.

\begin{rem}\label{rem-rate}{\rm
Let us discuss a link between the coefficient $\hat{c}=c\tau$ and the relaxation of the Markov chain towards the equilibrium distribution $\varpi$ can be also seen as follows. Denote
by $\Pi(\hat{c})$ the matrix of the Glauber Markov chain, where $\hat{c}= c\tau$ is considered
as a small parameter. By formulas \eqref{eq-detbal}, for a sufficiently small constant $\hat{c}_0 >0$ and for all $0 < \hat{c} < \hat{c}_0$
\begin{equation}
\label{eq-lazy}
\Pi(\hat{c}) = \left(1-\frac{\hat{c}}{\hat{c}_0}\right){\bf 1} + \frac{\hat{c}}{\hat{c}_0} \Pi(\hat{c}_0)\;,
\end{equation}
i.e., $\Pi(\hat{c})$ is an $\alpha$-lazy version of $\Pi(\hat{c}_0)$ with
$\alpha=\hat{c}/\hat{c}_0$, in the terminology of \cite{Fried}. It follows
from the spectral theory of Markov chains \cite[Chapter 12]{LP} that for small $\hat{c}$
all the eigenvalues of $\Pi(\hat{c})$ are positive, the maximal one equals $1$, and
the next to maximal is of the form $1-\gamma(\hat{c})$ with $\gamma(\hat{c})$, {\it the spectral gap}, being strictly positive. The spectral gap is responsible for the relaxation of
the Markov evolution towards the stationary distribution $\varpi$. It follows immediately
from \eqref{eq-lazy} that
$\gamma(\hat{c}) = (\hat{c}/\hat{c}_0)\gamma(\hat{c}_0)$,
i.e.,  {\it the spectral gap of the Glauber Markov chain is proportional to} $\hat{c}$.
}
\end{rem}

Consider a modification of the Glauber dynamics in which the magnetic field $q$ is continuously updated so that the effective magnetic field $Q$ remains constant. Let us explain how this modified dynamics is related to the second (contact) scenario in Section~\ref{subsec-ising-1}.
Namely, put \color{black} $\phi_{\beta}(x) = \beta^{-1} \log 2 (\cosh \beta x)$. \color{black}
 Define new coordinates $Q:= q+bp$ (effective magnetic field), \color{black}$P:= p- \phi'_{\beta}(Q)$\color{black}
(the deviation from the equilibrium magnetization, called also
{\it de Donder's  affinity}, see Section 4.1 in \cite{Haslach}).
With these $P,Q$, we rewrite equation \eqref{eq-evolution} above as
\begin{equation} \label{eq-Pc-vsp}
\dot{P}=-cP\;.
\end{equation}
Let us now choose a
coordinate $Z$ and a
contact Hamiltonian $H_0 (P,Q,Z)$ such that
the equations \eqref{eq-Pc-vsp} and $\dot{Q}=0$ appear in the system of three equations defined by $H_0$.
Such
$Z$ and
$H_0$ can be chosen as
$H_0(P,Q,Z):=-cZ$, where
$Z := z- \phi_{\beta}(Q) + bp^2/2$ (the deviation from the equilibrium free energy).
It is easy to see that $H_0$ is exactly the Hamiltonian
\eqref{eq-relax}
written in the coordinates $(P,Q,Z)$ defined here. (Note that $P,Z$ defined here are different from the ones used in Section~\ref{subsec-ising-1}).

\subsection{Perturbed Glauber dynamics}\label{subsec-pG}

Suppose now that we are given a perturbation $$H(P,Q,Z)= -cZ +F(P,Q,Z)$$ of $H_0$
so that $H$ is an admissible Hamiltonian in the sense of Definition~\ref{def-admissible-contact-Ham}.
The corresponding evolution is governed by the system
\begin{equation} \label{eq-system}
\begin{cases}
\dot{P} = -(c- \frac{\partial F}{\partial Z})P + \frac{ \partial F}{\partial Q} \\
\dot{Q}= -\frac{\partial F}{\partial P} \\
\dot{Z} = -cZ +F - P\frac{\partial F}{\partial P}
\end{cases}
\end{equation}
In the discussion below, we fix $F$ and assume that its derivatives are small compared to $c$. As it turns out
(see Remark \ref{rem-smallness} below), in this case
 the first equation of the system \eqref{eq-system}
 still admits an interpretation at the microscopic level via a perturbed Glauber dynamics, provided
the functions $Q$ and $Z$ are considered as time-dependent parameters.

More precisely, we shall make the following modifications (M1),(M2),(M3) to the continuous time Glauber dynamics described above:
\begin{itemize}
\item[{(M1)}] the time scale parameter equals
$c'= c- \frac{\partial F}{\partial Z}(P,Q,Z)$;
\item[{(M2)}] the transition probabilities \eqref{eq-w-1} (with $c$ replaced by $c'$)
are subject to an additive perturbation (noise) of the form
 $ - \frac{1}{2}\sigma(g) r (\sigma)$ satisfying \color{black}
\begin{equation}\label{eq-noise-vsp}
\mathbb{E}_t(r) = \frac{\partial F}{\partial Q}(P,Q,Z) - \phi''_{\beta} (Q) \frac{\partial F}{\partial P}(P,Q,Z)\;,
\end{equation}\color{black}
so that modified transition probabilities $w'_g(\sigma)$ are given by
\begin{equation}\label{eq-prob-modified}
w'_g(\sigma) =  \frac{c'}{2} \left(1- \sigma(g)\tanh \Big(\beta \big(q+ \sum_{h\neq g} J_{gh}\sigma(h)\big)\Big)\right)- \frac{1}{2}\sigma(g)r(\sigma)\;.
\end{equation}
\end{itemize}
Here $r$ is a small real-valued function of $\sigma$. It can be interpreted as a ``bias": when $r(\sigma)> 0$, spins $+1$ flip with smaller probability, and when $r (\sigma) <  0$, spins $-1$
flip with larger probability than in the unperturbed case.
We shall further assume that
\begin{itemize}
\item[{(M3)}] In the course of the dynamics, $Q$ and $Z$ are continuously updated according to the last two equations of \eqref{eq-system}.
\end{itemize}

\begin{rem}\label{rem-smallness}
{\rm Note that the time scale parameter
$c'$ should remain positive and the probabilities
$w'_g(\sigma)$ should lie in $[0,1]$. This holds true provided the first derivatives
of $F$ are sufficiently small compared to $c$
and the function $r$ is sufficiently small.
}
\end{rem}

We claim that under these assumptions the evolution of $P$ is given by the
first equation in the contact Hamiltonian system \eqref{eq-system}.
Indeed, modifying  \eqref{eq-evol-vsp} by using  \eqref{eq-prob-modified}
and applying the mean field assumption we get that
$$\dot{p} = -2\mathbb{E}_t(\sigma(g)w'_g(\sigma))= -c'P + \mathbb{E}_t(r)\;.$$
Substituting the expression for $c'$ from (M1) and using \eqref{eq-noise-vsp}
we obtain \color{black}
\begin{equation}\label{eq- Pp-vsp}
\dot{p}= -\left(c - \frac{\partial F}{\partial Z}\right)P + \frac{\partial F}{\partial Q} - \phi''_{\beta} (Q) \frac{\partial F}{\partial P}\;.
\end{equation}\color{black}
Recalling that by definition \color{black}$P = p - \phi'_{\beta}(Q)$\color{black}, and by the second equation in
\eqref{eq-system} $\dot{Q} = -\frac{\partial F}{\partial P}$, we get that \color{black}
$$\dot{P} = \dot{p} + \phi''_{\beta} (Q) \frac{\partial F}{\partial P}\;.$$\color{black}
Therefore, by \eqref{eq- Pp-vsp}
$$\dot{P}= -\left(c - \frac{\partial F}{\partial Z}\right)P + \frac{\partial F}{\partial Q}\;,$$ which is the first equation of \eqref{eq-system}. This proves the claim.

Now we readily get the following conclusions:

\bigskip
\noindent
1. {\it With the continuous time Glauber dynamics being adjusted according to (M1), (M2), (M3), the corresponding adjusted Glauber evolution in the mean field approximation  can be approximated by the contact Hamiltonian flow of $H$.}

\bigskip
\noindent
2. {\it Theorem~\ref{thm-ising} is applicable to $H$,
yielding a relaxation process defined by $H$ that starts at a point of $\Lambda_{a,\alpha}$ and asymptotically converges to $\Lambda_{b,\beta}$.}

\bigskip
Let us comment on the notion of admissibility for contact Hamiltonians introduced in Definition~\ref{def-admissible-contact-Ham}. Assumption (iv) means that the perturbed dynamics is close to  (the contact version of) the Glauber one near the equilibrium. An interpretation of assumptions (i)-(iii) is that away of the equilibrium the perturbation of the time scale parameter
as well as the noise are mild.
Periodicity in $Q$ seems to be technical (cf. Remark \ref{rem-per} above). It is unclear to us whether a mere boundedness of the perturbation $F$ (possibly, with several derivatives) would enable us to get an analogue of Theorem~\ref{thm-ising}.

It would be interesting to explore relaxation processes for Hamiltonians
of the simplified form $-cZ+F(P,Q)$, with $F=0$ near $P=0$. It is unclear to us whether the relaxation processes in this case can be detected by more elementary methods of two-dimensional dynamics in the $(P,Q)$-plane.

Let us mention also that an analysis of coupled Ising models (see e.g. \cite{Ising-coupled})
similar to the one we performed for the single model should lead to interesting asymptotic
questions of multi-dimensional contact dynamics in the space \color{black}$(\R^{2n+1}, dz- \sum p_i dq_i)$, \color{black}
where $z$ is the total free energy and $(p_i,q_i)$ are the magnetization and
the magnetic field for the $i$-th model. It is likely
that the methods developed in the present paper enable one to detect relaxation processes
in this more sophisticated model.

\bigskip
\noindent
{\bf Acknowledgment:} We thank Oleg Gendelman, Emmanuel Giroux, Leonid Levitov, Eva Miranda, Stefan Nemirovski, Ron Peled and Michael Shapiro for useful discussions. We are grateful to Alessandro Bravetti, Shin-itiro Goto, Shai Lerer, and Senya Shlosman for very helpful comments on the manuscript. Finally, we thank the anonymous referee for illuminating remarks
and useful suggestions, as well as for indicating a number of mistakes.

\bibliographystyle{alpha}

\end{document}